\documentclass[letterpaper,10pt]{article}

\usepackage[preprint]{probml}

\ShortHeadings{Intention Inference}

\usepackage{wrapfig}
\usepackage{booktabs}

\usepackage{bm,amssymb,amsmath}

\makeatletter
\def\th@plain{%
  \thm@notefont{}
  \itshape 
}
\def\th@definition{%
  \thm@notefont{}
  \normalfont 
}
\makeatother

\def\1{\bm{1}}

\DeclareMathAlphabet{\mathsfit}{\encodingdefault}{\sfdefault}{m}{sl}
\SetMathAlphabet{\mathsfit}{bold}{\encodingdefault}{\sfdefault}{bx}{n}

\begin{document}

\title{
  Intention Inference Under Execution Noise: Separating Aleatoric and Epistemic Uncertainty in Social Dilemmas
}

\author[a,c,$\dagger$]{Kival Mahadew}
\author[b,c,d,e]{Jonathan Shock}

\affil[a]{Dept.\ of Computer Science, University of Cape Town}
\affil[b]{Dept.\ of Mathematics \& Applied Mathematics, University of Cape Town}
\affil[c]{Neuroscience Institute, University of Cape Town}
\affil[d]{National Institute for Theoretical \& Computational Sciences, Stellenbosch}
\affil[e]{Institut National de la recherche scientifique, Montreal}
\affil[$\dagger$]{Correspondence to \url{kivalm@protonmail.com}}

\maketitle

\begin{abstract}
  In noisy social dilemmas, intended actions are stochastically corrupted before execution, so an observed defection may reflect hostile intent or action error. Standard Markov Decision Process (MDP) formulations treat executed actions as states, structurally precluding this distinction and causing systematic over-retaliation. We introduce a Partially Observable MDP (POMDP) formulation encoding opponent intentions as latent states and executed actions as noisy observations, solved within the active inference (AIF) framework with a cost function that decomposes into epistemic and pragmatic components that jointly address inferring current intent and learning how intent evolves. In the Iterated Prisoner’s Dilemma with symmetric noise, we derive a critical noise threshold governing cooperation collapse, connecting it to a fixed-point condition on learned priors. Experiments reveal that the value of intention inference is context-dependent: the POMDP provides consistent advantages against conditionally cooperative opponents, but mutual intention inference under sufficient noise produces correlated belief-driven collapse. The advantage is specific to games where intent attribution is decision-relevant.
\end{abstract}

\section{Introduction}
\label{sec:intro}
The emergence and stability of cooperation among self-interested agents is a foundational problem across economics, biology, and artificial intelligence. The canonical formalisation is the Iterated Prisoner's Dilemma (IPD), where conditional strategies can sustain mutual cooperation provided agents reliably observe each other's choices. Under execution noise, where intended actions are stochastically flipped before enactment, this condition fails. An agent observing defection faces an attribution problem: deliberate hostility, or unintended error?

This attribution problem is more consequential than it first appears, because noise does not merely add variance to outcomes but restructures the learning dynamics themselves. Agents that respond to corrupted observations as though they were deliberate create non-stationarity for their partners, who adapt in turn, producing a volatile feedback loop in which noise, misattribution, and adaptation compound one another. The classical literature addresses this through engineered forgiveness: Generous Tit-for-Tat cooperates probabilistically after observed defection, effectively averaging over causes. While practically effective, such approaches lack a principled account of how much forgiveness the evidence warrants and cannot adapt to the specific opponent. Reinforcement learning offers generality, but standard MDP formulations represent the IPD with executed actions as state, providing no room for latent intention variables. The agent structurally cannot form the hypothesis 'my opponent intended to cooperate but executed incorrectly,' leading to systematic over-retaliation.

We argue the correct formulation is a POMDP where hidden states encode latent intentions and observations correspond to noisy executed actions. Belief updating then constitutes probabilistic inference about what the opponent meant, conditioned on what they did, giving us a principled, Bayesian account of forgiveness. We solve this within Active Inference (AIF), whose expected free energy decomposition provides natural machinery for the dual problem of inferring current intent and discovering how intent evolves: epistemic value drives information-seeking actions that resolve strategic uncertainty, while pragmatic value drives goal-directed exploitation of learned structure.

We make four contributions. (1) A POMDP formulation separating aleatoric execution noise from epistemic uncertainty about intent, with preferences over observations consistent with the payoff structure. (2) Online learning of the opponent's reactive policy via AIF, coupling state inference with model learning through Dirichlet count accumulation. (3) An analysis of the interacting mechanisms that give rise to a noise threshold beyond which mutual intention inference produces correlated failure rather than cooperation. (4) Empirical demonstration across opponent types revealing that the value of intention inference is context-dependent: beneficial against reciprocating opponents, fragile under mutual inference at high noise, and costly against exploitable opponents.

\section{Related Work}
\textbf{Cooperation under execution noise.}
The destabilising effect of execution noise on reciprocal cooperation has been studied since
\cite{molanderOptimalLevelGenerosity1985} showed that Tit-for-Tat collapses under arbitrarily small error rates. Subsequent work introduced fixed probabilistic forgiveness~\citep{glynatsiPropertiesWinningIterated2024}, internal standing variables distinguishing execution from perception errors~\citep{boerlijstLogicContrition1997}, symbolic pattern matching over accident versus intention~\citep{auAccidentIntentionThat2006}, and memory-1 strategies jointly satisfying cooperation, error correction, and retaliation~\citep{hilbeMemorynStrategiesDirect2017}. Our POMDP agent generalises these fixed-memory approaches: belief states provide adaptive effective memory, and Bayesian updating replaces engineered forgiveness with a graded posterior over cooperative intent.

\textbf{Opponent modelling and intent inference.}
Interactive POMDPs~\citep{gmytrasiewiczFrameworkSequentialPlanning2005} embed recursive opponent models into the state space, yielding optimal but intractable planning; our Dirichlet posterior over the opponent's reactive policy is a computationally tractable, non-recursive special case. Bayesian Theory of Mind~\citep{bakerActionUnderstandingInverse2009} and neural meta-learning approaches~\citep{rabinowitzMachineTheoryMind2018} both treat opponent mental states as hidden causes of observed behaviour, but neither separates execution noise from strategic intent. \cite{foersterLearningOpponentLearningAwareness2018} study learning dynamics under co-adaptation through differentiable opponent models; our self-play collapse provides a complementary perspective from the Bayesian attribution side.

\textbf{Active inference in strategic settings.}
\cite{yoshidaGameTheoryMind2008} first connected variational inference to game-theoretic interaction, and \cite{fristonDuetOne2015} cast Theory of Mind as active inference over a shared generative model. \cite{demekasAnalyticalModelActive2023} provide an analytical MDP model of the IPD from which our MDP-AIF baseline follows, deriving phase transition conditions for cooperation but without execution noise. \cite{ruiz-serraFactorisedActiveInference2024} introduce factorised generative models with explicit beliefs about opponents' internal states in iterated general-sum games. Our work complements these by making execution noise the central challenge and showing that the representational choice between latent intentions and observed actions has qualitative consequences for cooperation and inference stability.

\textbf{Uncertainty decomposition.}
The aleatoric/epistemic distinction is well-established in supervised learning~\citep{kendallWhatUncertaintiesWe2017} and has been applied to risk-sensitive RL through information-theoretic decompositions~\citep{depewegDecompositionUncertaintyBayesian2018}. \cite{ghoshWhyGeneralizationRL2021} show that epistemic uncertainty transforms MDPs into implicit POMDPs requiring information-gathering actions. In the multi-agent setting, \cite{diaconescuInferringIntentionsOthers2014} use hierarchical Bayesian learning to separate observation noise from an adviser's underlying intentions, achieving an analogous decomposition in a social learning task. Our contribution connects these threads: the first POMDP formulation combining the game-theoretic recognition that noise types have different strategic consequences, with Bayesian uncertainty separation and active inference planning.

\section{The Latent Intention POMDP}
\label{sec:pomdp}

We model execution noise as a binary symmetric channel with parameter $\epsilon \in [0, 0.5)$: each player's intended action is independently flipped with probability $\epsilon$. We use the standard IPD payoffs $T{=}5, R{=}3, P{=}1, S{=}0$.

The POMDP $\langle \mathcal{S}, \mathcal{O}, \mathcal{A}, \mathbf{A}, \mathbf{B}, \mathbf{C} \rangle$ is defined as follows. \textbf{States} $\mathcal{S} = \{\text{Start}, CC, CD, DC, DD\}$ encode joint intended actions prior to execution; Start is occupied only at $t{=}0$. \textbf{Observations} $\mathcal{O} = \{\text{start}, CC, CD, DC, DD\}$ are the executed joint actions. \textbf{Actions} $\mathcal{A} = \{C, D\}$ are the focal agent's intended action.

The \textbf{observation model} ($\mathbf{A}$ matrix) $P(o \mid s)$ is fixed and known, parameterised by $\epsilon$. Each player's intended action is independently flipped, yielding: $P(\text{correct joint action} \mid s) = (1{-}\epsilon)^2$, $P(\text{one flip} \mid s) = \epsilon(1{-}\epsilon)$, $P(\text{both flipped} \mid s) = \epsilon^2$.

The \textbf{transition model} ($\mathbf{B}$ matrix) $P(s_{t+1} \mid s_t, a_t)$ captures the opponent's reactive policy. Here $s_t = (a_{t-1}, \iota^{\text{opp}}_t)$ is the joint intention state at round $t$ and $a_t \in \{C, D\}$ is the focal agent's intended action for round $t{+}1$, selected after observing $o_t$. The stochastic content of this transition is entirely the opponent's next intention $\iota^{\text{opp}}_{t+1}$, since the agent's component of $s_{t+1}$ is deterministically $a_t$. This yields the factorisation $P(s_{t+1} \mid s_t, a_t) = P(\iota^{\text{opp}}_{t+1} \mid s_t, a_t) \cdot \mathbf{1}[\iota^{\text{self}}_{t+1} = a_t]$. Due to simultaneous play, the opponent's intention at $t{+}1$ depends on the joint outcome at $t$ (a noisy function of $s_t$) but cannot depend on $a_t$, since $o_t$ is determined before $a_t$ is chosen; the agent's generative model retains the $(s, a)$ indexing because $a_t$ determines which pair of next states ($\{a_t C, a_t D\}$) is reachable. The $\mathbf{B}$ matrix is unknown and learned online (\cref{sec:online-learning}).

The \textbf{preference vector} $\mathbf{C}(o) = \ln \tilde{P}(o)$ assigns log preferences proportional to payoffs over observations: $C(DC) > C(CC) > C(DD) > C(CD)$, reflecting that payoffs are determined by executed actions, not intentions.

\section{Active Inference Solution}
\label{sec:aif}

\subsection{Background}
\label{sec:aif-background}

Active inference is a Bayesian framework, originating in theoretical neuroscience, that casts
perception, action, and learning as minimising uncertainty under a generative model of the
environment \citep{fristonActiveInferenceCuriosity2017}. An agent maintains a probabilistic model of how hidden
states produce observations and how actions drive state transitions, then selects actions that
jointly resolve uncertainty about the world and steer observations toward preferred outcomes.
Inference over hidden states minimises the \emph{variational free energy} $F$, the negative
evidence lower bound (ELBO). Action selection minimises its forward-looking counterpart, the
\emph{expected free energy}, whose decomposition into epistemic and pragmatic components provides
a principled balance between exploration and exploitation without requiring a separate
exploration mechanism.

This framework suits the noisy IPD for two reasons. First, the generative model naturally
accommodates the POMDP structure of \cref{sec:pomdp}: hidden intention states, a known
observation model encoding execution noise, and an unknown transition model encoding the
opponent's policy. Second, the expected free energy's epistemic terms provide exactly the
machinery needed to actively learn the opponent's reactive policy through noise, rather than
passively accumulating experience.

\subsection{Generative Model}
\label{sec:generative-model}

The generative model specifies the joint distribution:
\[
  P(\tilde{o}_{1:T}, \tilde{s}_{1:T}, \pi, \mathbf{B})
  = P(\pi)\, P(\mathbf{B}) \prod_{t=1}^{T} P(o_t \mid s_t, \mathbf{A})\, P(s_t \mid s_{t-1}, a_{t-1}, \mathbf{B})
\]
where $\pi$ denotes a policy and $\mathbf{A}$ is fixed. The approximate posterior factorises as
$Q(\tilde{s}, \pi, \mathbf{B}) = Q(\pi)\, Q(\mathbf{B}) \prod_t Q(s_t)$.
\Cref{app:notation} maps the active-inference symbols used here (the likelihood $\mathbf{A}$, transition $\mathbf{B}$, posterior $Q$, and free energies $F$ and $G$) to their standard probabilistic-ML analogues.

\subsection{Belief Updating and Policy Selection}
\label{sec:belief-updating}

Beliefs update via Bayesian filtering:
\[
  Q(s_t) \propto P(o_t \mid s_t, \mathbf{A}) \sum_{s_{t-1}} P(s_t \mid s_{t-1}, a_{t-1}, \mathbf{B})\, Q(s_{t-1})
\]
This implements the key inference: given observed defection, noise $\epsilon$, and prior beliefs about the opponent, what is the posterior probability of cooperative intent?

Policy selection minimises the expected free energy $G(\pi)$, which decomposes at future timestep $\tau$ into pragmatic and epistemic terms
\begin{align}
  G(\pi, \tau)
   & = \underbrace{-\mathbb{E}_{Q(o_\tau, s_\tau | \pi)}\!\left[\ln P(o_\tau \mid s_\tau)
  - \ln Q(s_\tau \mid \pi)\right]}_{\text{State information gain}} \nonumber                 \\
   & \quad \underbrace{- \mathbb{E}_{Q(o_\tau | \pi)}\!\left[\ln \tilde{P}(o_\tau)\right]}_{
  \text{Pragmatic value}} \nonumber                                                          \\
   & \quad \underbrace{- \mathbb{E}_{Q(o_\tau, s_\tau | \pi)}\!\left[D_{\mathrm{KL}}\!\left[
        Q(\mathbf{B} \mid o_\tau, s_\tau) \,\|\, Q(\mathbf{B})\right]\right]}_{
    \text{Parameter information gain}}
  \label{eq:efe}
\end{align}

\textbf{Pragmatic value} drives the agent toward preferred (high-payoff) observations.
\textbf{State information gain} measures how much an action would reduce uncertainty about
the current hidden state, favouring actions whose outcomes are most diagnostic of latent
intentions. \textbf{Parameter information gain} measures how much an action would reduce
uncertainty about the transition model, favouring actions that are informative about the
opponent's reactive policy. The interplay between these terms and the two formulations is
examined empirically in \cref{sec:efe-decomp}.

\subsection{Online Learning of Transition Dynamics}
\label{sec:online-learning}

The $\mathbf{B}$ matrix is learned online via Dirichlet priors over each conditional
$P(s' \mid s, a)$, parameterised by observation counts $N^a_{ss'}$ under a symmetric
$\alpha_0 = 1$ (uniform) prior. The estimate is the Dirichlet posterior mean, equivalently the
posterior predictive for the next state, which applies add-one (Laplace) smoothing to the counts:
\[
  \hat{P}(s' \mid s, a) = \frac{N^a_{ss'} + \alpha_0}{\sum_{s''} \bigl(N^a_{ss''} + \alpha_0\bigr)},
  \qquad \alpha_0 = 1.
\]
Counts are updated at a tunable interval
$\Delta t$ optimised via hyperparameter search. This constitutes online Bayesian inference over
the opponent's reactive policy, coupled with state inference: accurate intention estimates
provide a clean training signal for the Dirichlet counts, and an accurate $\mathbf{B}$ matrix
in turn improves state predictions.

\section{Critical Noise Threshold}
\label{sec:threshold}

We derive the conditions under which cooperation survives noisy execution in the POMDP formulation. The opponent's cooperative prior from state $s$ at timestep $t$ is the Dirichlet posterior mean of the learned counts, using the same add-one convention as \cref{sec:online-learning}:
\[
  p_t(s) = \frac{N^C_{s \to CC} + \alpha_0}{\sum_{s'} \bigl(N^C_{s \to s'} + \alpha_0\bigr)},
  \qquad \alpha_0 = 1
\]
An agent with planning horizon $h$ does not evaluate a single observation against this prior but aggregates evidence over $h$ rounds, asking whether the observed pattern of defections is better explained by noise or genuine hostility. We formalise this as a binary hypothesis test. Suppose the agent observes $k$ opponent defections in $h$ consecutive rounds from state CC. Under the cooperative hypothesis $H_C$, the opponent's intention remains cooperative throughout the window with probability $p_t(\text{CC})^h$, and observed defections arise solely from the execution channel: $k \sim \text{Bin}(h, \varepsilon)$. Under the hostile hypothesis $H_D$, with probability $1 - p_t(\text{CC})^h$, observed defections arise at rate $1 - \varepsilon$.

\begin{proposition}[$h$-step forgiveness condition]
  \label{prop:hstep}
  For planning horizon $h$, noise level $\varepsilon$, and cooperative prior $p = p_t(\text{CC})$, the maximum number of observed defections attributed to noise is
  \begin{equation}
    k^*(h, p, \varepsilon)
    \;=\;
    \left\lfloor
    \frac{1}{2}\!\left(
    h \;+\;
    \frac{\ln\!\bigl(p^h / (1 - p^h)\bigr)}{\ln\!\bigl((1-\varepsilon)/\varepsilon\bigr)}
    \right)
    \right\rfloor
    \label{eq:kstar}
  \end{equation}
    The two terms inside the floor separate the competing effects: the $h/2$ term is the symmetric midpoint at which noise alone leaves the evidence balanced, while the second sets the prior cooperative-survival odds $p^h/(1-p^h)$ against the per-observation discriminability $\ln\bigl((1-\varepsilon)/\varepsilon\bigr)$ of the channel.
  At $h = 1$, this recovers the myopic attribution boundary: $k^* \geq 1$ iff $p > 1 - \varepsilon$, which is equivalent to requiring $P(\text{int.}\ C \mid \text{obs.}\ D) > 0.5$ under the binary symmetric channel (\cref{app:myopic-boundary}).
\end{proposition}

The proof is given in \cref{app:self-play-collapse}. Two quantities govern whether cooperation is sustained, and they can diverge. The first is \emph{discriminability}, measured by $k^*$: how many noisy defections the agent can correctly attribute to the execution channel within a single planning window. Since $k^*$ is non-decreasing in $h$, longer planning horizons always improve the agent's ability to tolerate individual noisy observations. The second is the \emph{cooperative survival probability} $p_t(\text{CC})^h$: the probability that cooperative intent persists across the full planning window. Because $p_t(\text{CC}) < 1$, this quantity decays exponentially in $h$, so longer horizons make it less likely that the cooperative hypothesis applies in the first place.

The critical threshold is the noise level at which the cooperative survival probability crosses below $0.5$ for the agent's planning horizon: the point where the hostile hypothesis becomes the typical case over the planning window regardless of discriminability. When this occurs, the agent retaliates not because it cannot distinguish noise from hostility but because genuinely hostile intent is the more probable explanation. The forgiveness machinery remains accurate but addresses the minority case. \Cref{sec:attribution-verification} verifies this prediction empirically, and \cref{app:self-play-collapse} traces the belief-driven cascade by which noise-induced retaliation becomes self-reinforcing under mutual inference.

This binary partition reflects the game's strategic structure: rational play in the IPD drives behaviour toward the symmetric equilibria (CC or DD), and the bistable equilibrium structure (\cref{app:self-play-collapse}) concentrates converged dynamics at these extremes. A full mixture over cooperative-count hypotheses (\cref{app:mixed-intent}) yields a weakly lower threshold, but the gap is zero at every noise level where cooperation is sustained and at most one defection where it has already collapsed.



\section{Experiments}
\label{sec:experiments}

\subsection{Setup}
\label{sec:setup}

We evaluate POMDP-AIF (proposed) against MDP-AIF (standard active inference with executed actions as states, following the formulation of \citet{demekasAnalyticalModelActive2023}) across three conditions: independent self-play, vs Tit-for-Tat (TFT), and vs Win-Stay Lose-Shift (WSLS). Noise levels: $\epsilon \in \{0.0, 0.05, 0.10, 0.15, 0.20, 0.25\}$. The A matrix is fixed; the B matrix is initialized with uniform Dirichlet prior and updated at interval $\Delta t$ tuned on self-play and vs TFT\@. Results are averaged over 30 seeds, 5000 timesteps per episode.

\subsection{Cooperation and Reward Under Noise}
\label{sec:coop-reward}

\begin{figure}[t]
  \centering
  \includegraphics[width=\linewidth]{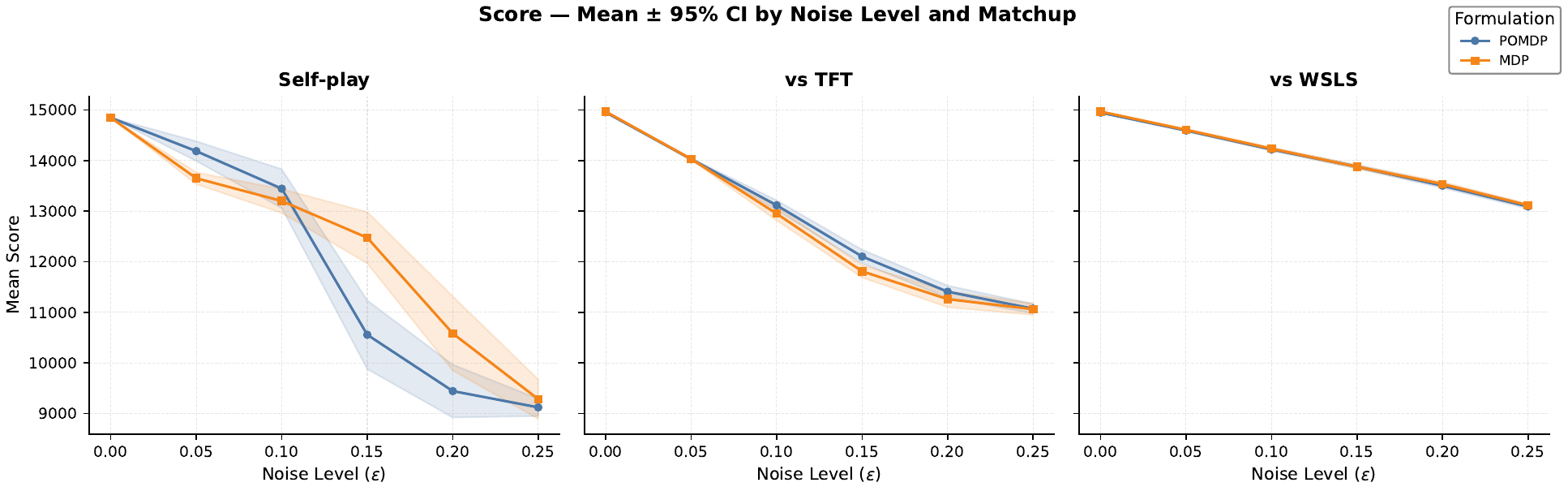}
  \caption{Mean cumulative score ($\pm$ 95\% CI, 30 seeds, 5000 timesteps) as a function of noise level $\epsilon$. Left: self-play. Centre: vs TFT\@. Right: vs WSLS\@. The POMDP outperforms at low noise in self-play but collapses sharply between $\epsilon = 0.10$ and $0.15$; against TFT, it maintains a consistent advantage across all noise levels.}
  \label{fig:scores}
\end{figure}

\begin{figure}[t]
  \centering
  \includegraphics[width=\linewidth]{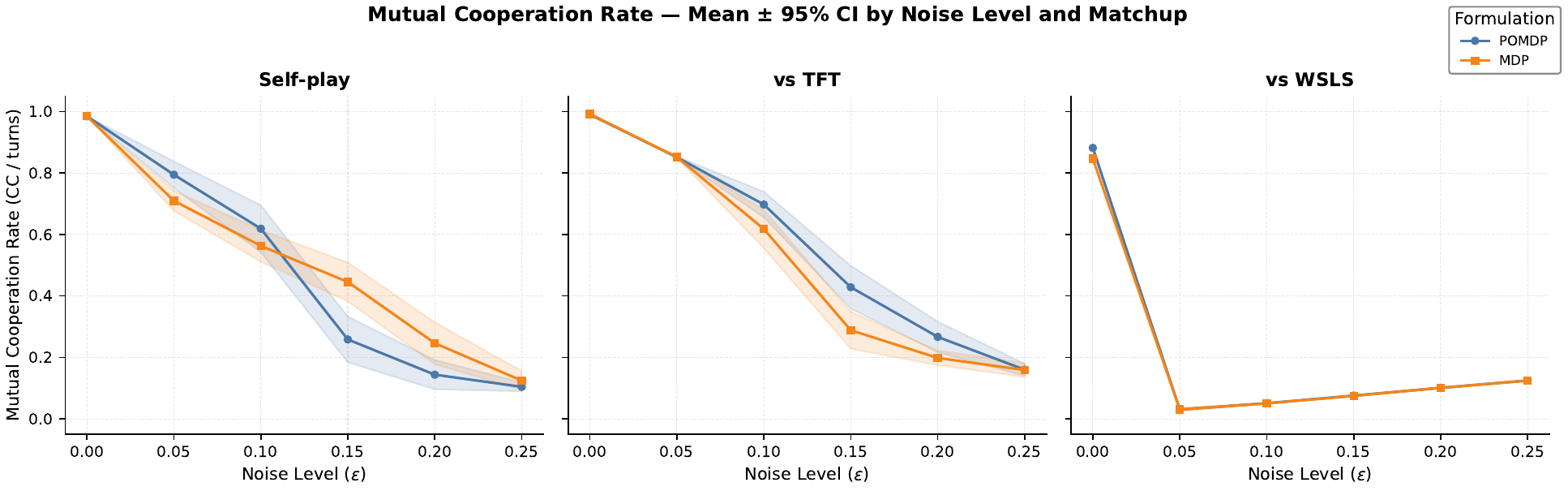}
  \caption{Mutual cooperation rate (fraction of CC outcomes, $\pm$ 95\% CI) as a function of $\epsilon$. Same layout as \cref{fig:scores}. Against WSLS, cooperation collapses at any $\epsilon > 0$ while scores remain high (\cref{fig:scores}), indicating exploitation rather than cooperation breakdown.}
  \label{fig:coop-rate}
\end{figure}

\textbf{Self-play.}
Both formulations achieve near-perfect cooperation in the noiseless setting (\cref{fig:scores,fig:coop-rate}).
At low noise ($\epsilon = 0.05$), the POMDP maintains a clear cooperation advantage, consistent with its ability to attribute observed defection to execution error rather than hostile intent.
The two formulations converge near $\epsilon = 0.10$, after which the POMDP undergoes a sharp collapse: between $\epsilon = 0.10$ and $0.15$, cooperation falls to $0.25$ and cumulative score drops below MDP levels.
The MDP degrades gradually across the full noise range.
This asymmetry reflects a qualitative difference in failure mode: the POMDP's collapse is driven by correlated belief dynamics in mutual inference (\cref{sec:attribution-verification}), while the MDP's smooth decline reflects the absence of belief machinery that could produce such a transition.
Wide confidence intervals at intermediate noise levels reflect the bimodal character of this transition across seeds.

\textbf{vs TFT.}
Against a stationary, conditionally cooperative opponent, the POMDP maintains a consistent advantage across all noise levels with no catastrophic collapse (\cref{fig:scores,fig:coop-rate}). TFT's fixed policy provides a stable learning target that cannot be destabilised by the focal agent's belief dynamics, so the correlated failure mode is structurally absent. Both formulations converge at high noise ($\epsilon = 0.25$) as execution noise overwhelms the cooperative signal regardless of representation.

\textbf{vs WSLS.}
Against WSLS, cooperation collapses at any $\epsilon > 0$ for both formulations, yet cumulative scores remain high (\cref{fig:scores,fig:coop-rate}). Both agents learn to exploit WSLS's payoff-responsive switching rule, generating a DC/DD cycle whose average payoff matches mutual cooperation. The POMDP's cooperative prior delays this transition to exploitation, so the MDP achieves marginally higher scores. This condition tests opponent-model adaptation speed rather than cooperation capacity, and shows that the POMDP's forgiving disposition becomes a liability against exploitable opponents.

\textbf{Scope of the advantage.} Two conditions bound this advantage: it vanishes without the epistemic terms (\cref{app:no-epistemic}), and it shrinks if intent attribution lacks decision-relevance, as in the Stag Hunt (\cref{app:stag-hunt}). \Cref{sec:discussion} details both.

\subsection{EFE Decomposition and Learning Dynamics}
\label{sec:efe-decomp}

\begin{figure}[t]
  \centering
  \includegraphics[width=\linewidth]{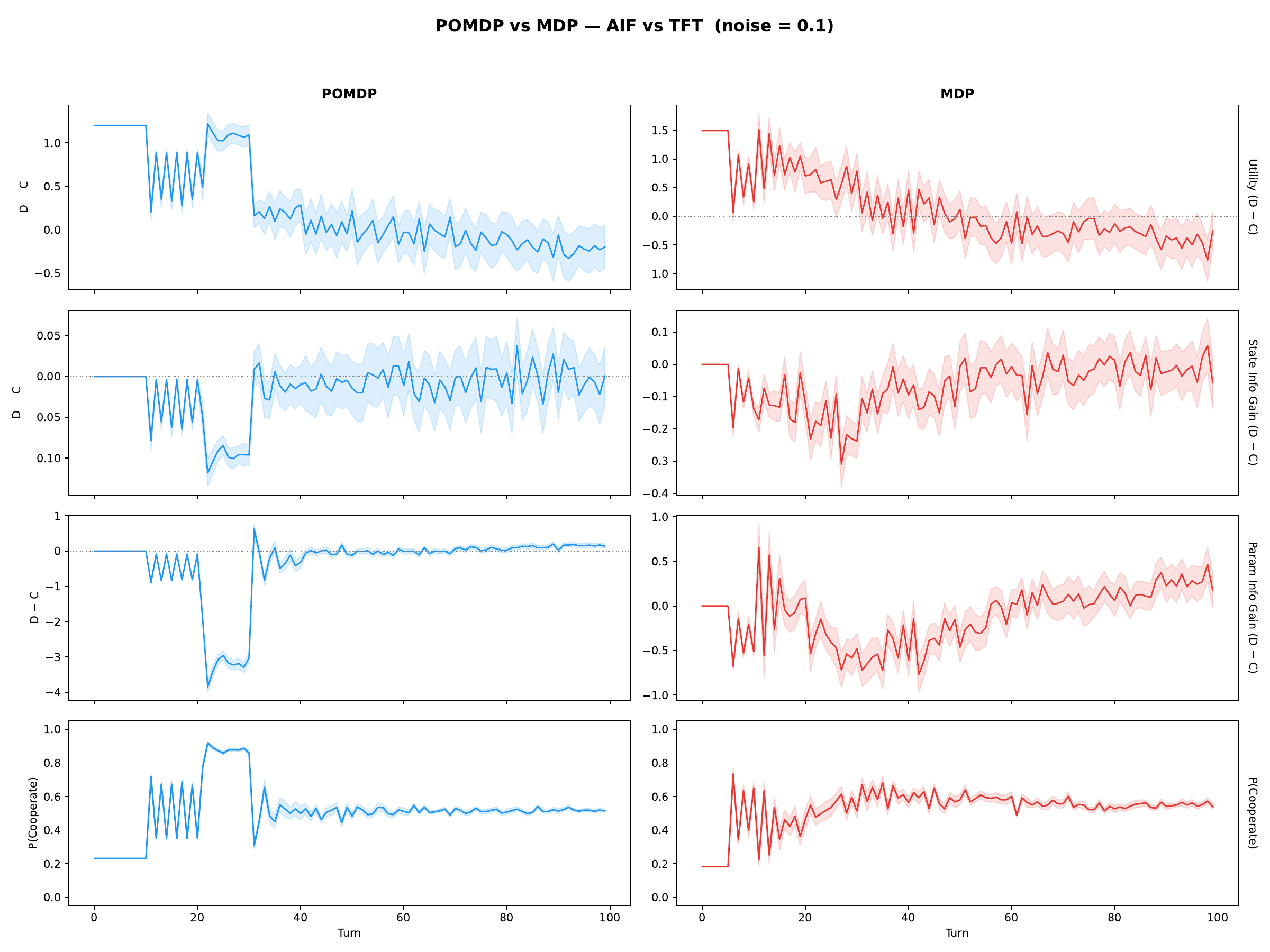}
  \caption{Expected free energy decomposition for POMDP-AIF (left, blue) and MDP-AIF (right, red) vs TFT at $\epsilon = 0.10$, averaged over 30 seeds with 95\% confidence intervals for the first 100 turns. Each row shows the difference D$-$C (positive favours defection) for one EFE component: pragmatic value (row 1), state information gain (row 2), parameter information gain (row 3), and the resulting P(Cooperate) (row 4).}
  \label{fig:efe-decomp}
\end{figure}

To understand what drives the divergence between formulations visible in \cref{fig:scores,fig:coop-rate}, we decompose the expected free energy into its constituent terms. \Cref{fig:efe-decomp} shows the D$-$C difference (positive favours defection) for each EFE component against TFT at $\epsilon = 0.10$, averaged over 30 seeds with 95\% confidence intervals.

\paragraph{Explore-then-commit in the POMDP.}
The POMDP agent exhibits a concentrated learning phase followed by stable commitment. Parameter information gain (row 3, left) drops sharply between turns 15 and 25, indicating that cooperation is far more informative about the transition model than defection during this window. This epistemic signal drives P(Cooperate) to $\approx 0.9$ (row 4). Once the Dirichlet counts converge, parameter information gain returns to zero and the agent settles into a cooperative policy sustained by pragmatic value alone.

\paragraph{Persistent indecision in the MDP.}
The MDP shows no comparable transition. Pragmatic value (row 1, right) initially favours defection and decays only gradually. State information gain (row 2, right) fluctuates throughout without stabilising, reflecting an unresolvable drive to clarify the current state. Parameter information gain (row 3, right) never produces the clean convergence signal that would enable policy commitment. The resulting P(Cooperate) drifts upward slowly, settling only as Dirichlet counts accumulate enough to overcome the initial defection bias.

\paragraph{Root cause.}
In the POMDP, the known observation model absorbs execution noise at the likelihood level, leaving latent intentions as a clean inference target. State information gain remains near zero throughout (row 2, left), and uncorrupted belief updates provide a reliable training signal for the transition model. In the MDP, observations serve directly as states, so aleatoric noise manifests as irreducible state uncertainty. State information gain never resolves to zero regardless of accumulated evidence, and this irreducibly elevated epistemic signal corrupts the Dirichlet counts, slowing convergence and preventing policy commitment. More generally, when epistemic drive fails to resolve in a stationary environment, the state representation is likely conflating reducible and irreducible uncertainty.

\subsection{Empirical Verification of the Noise Threshold}
\label{sec:attribution-verification}

To validate the critical noise threshold derived in \cref{sec:threshold}, we extract the Dirichlet cooperative prior $p_{500}(\text{CC})$ at $t = 500$ (10\% of the trajectory) and evaluate \cref{prop:hstep} at the self-play planning horizon $h = 5$. All priors are averaged across 30 seeds.

\begin{table}[h]
  \centering
  \caption{Multi-step attribution metrics at $h = 5$ (POMDP, self-play, $t = 500$). $k^*$: maximum forgivable defection count. $P(k \leq k^* \mid H_C)$ and $P(k \leq k^* \mid H_D)$: probability of observing a forgivable pattern under cooperative and hostile hypotheses.}
  \label{tab:multistep-verification}
  \begin{tabular}{cccccc}
    \toprule
    $\varepsilon$ & $p_{500}(\text{CC})$ & $p_{500}(\text{CC})^5$ & $k^*$ & $P(k \leq k^* \mid H_C)$ & $P(k \leq k^* \mid H_D)$ \\
    \midrule
    0.05          & 0.954                & 0.790                  & 2     & 99.9\%                   & 0.1\%                    \\
    0.10          & 0.907                & 0.614                  & 2     & 99.1\%                   & 0.9\%                    \\
    0.15          & 0.761                & 0.255                  & 2     & 97.3\%                   & 2.7\%                    \\
    \bottomrule
  \end{tabular}
\end{table}

The data confirm the central prediction of \cref{sec:threshold}: cooperation collapses because the cooperative survival probability degrades, not because discriminability fails. All three noise levels yield identical $k^*$ and discriminability gaps above 94 percentage points (\cref{tab:multistep-verification}), so the agent's ability to distinguish noise from hostility does not deteriorate across the phase transition.

What changes is the cooperative survival probability. At low noise the planning window is predominantly cooperative, and the empirical cooperation rate is high (\cref{fig:coop-rate}). As noise increases, $p_{500}(\text{CC})^5$ falls below $0.5$ and the hostile hypothesis dominates the planning horizon. The agent retaliates not because it confuses noise for hostility but because genuinely hostile intent is the more probable explanation over its planning window, leading to rational retaliation.

The MDP shows no comparable threshold because it lacks the latent intention layer over which \cref{prop:hstep} is defined. Without the capacity to form the cooperative and hostile hypotheses, the MDP degrades gradually through the mechanism described in \cref{sec:efe-decomp}. \Cref{app:self-play-collapse} provides the full state-dependent analysis, deriving the myopic attribution boundary as the $h = 1$ special case and tracing the belief-driven cascade by which noise-induced retaliation becomes self-reinforcing under mutual inference.

\section{Discussion and Conclusion}
\label{sec:discussion}

\paragraph{Representation and epistemic drive are jointly necessary.}
The ablation in \cref{app:no-epistemic} confirms that neither component suffices alone: removing epistemic terms eliminates the explore-then-commit trajectory even with the correct latent structure, while the MDP with full epistemic drive exhibits persistent exploratory pressure that never resolves into policy commitment. The cooperative advantages we report arise specifically from their interaction: the POMDP confines aleatoric uncertainty to the likelihood, creating an epistemically resolvable latent space, and the expected free energy's information gain terms drive the active probing needed to learn the opponent's policy through that space. This interdependence is the reason the explore-then-commit trajectory (\cref{sec:efe-decomp}) appears only in the POMDP with full EFE, and not in either component alone. The Stag Hunt results (\cref{app:stag-hunt}) reveal that this advantage is specific to games where the temptation payoff exceeds the reward for mutual cooperation ($T > R$): when there is no exploitation incentive, intent attribution is less decision-relevant and both formulations perform similarly.

\paragraph{The self-play collapse is a property of mutual inference, not of the formulation.}
The sharp transition between $\varepsilon = 0.10$ and $0.15$ is entirely absent against fixed opponents, confirming that it arises from correlated belief dynamics rather than from a deficiency in the POMDP model. The same sensitivity to inferred intent that enables near-perfect cooperation at low noise becomes a liability when noise-induced retaliation is correctly identified by both agents, making the spiral self-reinforcing. This is not a failure of inference but a consequence of its accuracy under mutual deployment, and it places a principled bound on the conditions under which intention inference improves cooperation.

\paragraph{Limitations and scope.}
The POMDP is given the correct noise model (the A matrix parameterised by the true $\varepsilon$), while the MDP has no access to this information. The comparison therefore involves both a structural advantage (the latent intention space) and an informational advantage (knowledge of $\varepsilon$). We argue the structural advantage is the more fundamental factor: without a latent intention layer there is no variable on which to condition, so knowledge of $\varepsilon$ alone would be inoperative. The A matrix is useful precisely because the POMDP has the representational structure to exploit it. A more controlled comparison including a noise-aware MDP baseline that discounts observed defections heuristically using the known $\varepsilon$ is left to future work. Such baselines exhibit various failure modes: for example, an MDP with a noise-representative prior recovers much of the POMDP's advantage but biases learning and underestimates the true noise level. The ablation in \cref{app:no-epistemic} partially isolates this by removing the POMDP's epistemic terms to separate the representational and informational contributions. Furthermore, a finer hyperparameter sweep (\cref{app:hyperparams}) may yield additional gains.

Our evaluation covers two-player, binary-action games with symmetric noise. The small state space permits exact inference; scaling to richer intention models would require approximate methods such as Monte Carlo tree search, since the POMDP agent's planning cost scales as $|A|^h$ for open-loop policy enumeration. The Dirichlet learning assumes stationary opponent policies, and the self-play results demonstrate the consequences when this assumption is violated: both agents' policies co-evolve, producing a form of model misspecification that compounds through mutual adaptation. Forgetting mechanisms or change-point detection could mitigate this cascade by allowing the cooperative prior to recover after transient retaliatory episodes. Two natural extensions remain. First, modeling non-stationary opponents requires opponent-shaping objectives or online inference of a memory-$n$ opponent, rather than assuming a fixed reactive policy. Second, a hierarchical POMDP could handle unknown or time-varying noise $\varepsilon$ by placing a prior over the observation model to infer $\varepsilon$ jointly with intention, replacing the fixed $\mathbf{A}$ matrix. More broadly, when epistemic drive fails to resolve in a stationary environment, the state representation is likely conflating reducible and irreducible uncertainty, a diagnostic that may apply to any domain where agents must attribute observed behaviour to latent causes through a noisy channel.

\paragraph{}
The central message is that the representational choice of what to treat as state versus observation determines whether principled cooperation under noise is structurally possible, and that the power of the resulting inference is bounded by the context in which it operates.

\clearpage

\section*{Acknowledgements}
This work was supported by funding from Conscium.

\bibliography{main}

\appendix
\section{Hyperparameter Configuration}
\label{app:hyperparams}

\subsection*{Active Inference Agent: Search Procedure and Selected Values}

\subsubsection*{A.1 Overview}

Hyperparameters were selected via an exhaustive grid search over the parameters and ranges listed in Table~\ref{tab:search-space}. Each candidate configuration was evaluated in self-play and against Tit-for-Tat across noise levels of 0.00 to 0.35 (step 0.05), with 5 independent repetitions of 1,000 turns per noise level. For each experimental context, the configuration achieving the highest mean score across all noise levels was retained for the main experiments. The full grid is reported in \cref{tab:grid-score-self,tab:grid-score-tft,tab:grid-cc-self,tab:grid-cc-tft} (Appendix A.5).

Two structural variants were explored: a POMDP formulation in which the observation model includes the known noise component, and the MDP formulation in which observations are treated as noiseless. Separate optima were identified for each variant and for each experimental context (self-play vs.\ cross-play), yielding four distinct configurations (Table~\ref{tab:selected-hyperparams}).

Parameters not included in the grid search were fixed at sensible defaults. The full combinatorial search space comprised \textbf{32 candidate configurations} (4 policy lengths $\times$ 4 update intervals $\times$ 2 observation model variants), with all other parameters held constant.

\subsubsection*{A.2 Search Space}

\begin{table}[h]
  \centering
  \caption{Hyperparameter search space.}
  \label{tab:search-space}
  \begin{tabular}{lcc}
    \toprule
    Parameter                         & Values Searched & Levels \\
    \midrule
    Planning horizon (steps)          & 1, 3, 5, 10     & 4      \\
    B matrix update frequency (steps) & 5, 10, 20, 50   & 4      \\
    POMDP vs.\ MDP formulation        & False, True     & 2      \\
    \bottomrule
  \end{tabular}
\end{table}

\noindent\textbf{Search environment:} noise levels 0.00 to 0.35 (step 0.05, 8 levels), 5 repetitions, 1,000 turns, seed 42.

\subsubsection*{A.3 Selected Hyperparameters}

\begin{table}[h]
  \centering
  \caption{Hyperparameters selected for the main experiments.}
  \label{tab:selected-hyperparams}
  \begin{tabular}{lccc}
    \toprule
    Experiment                 & POMDP & Planning horizon (steps) & B matrix update frequency \\
    \midrule
    Exp.\ 1 (POMDP, self-play) & True  & 5                        & 50                        \\
    Exp.\ 2 (MDP, self-play)   & False & 3                        & 50                        \\
    Exp.\ 3 (POMDP vs.\ TFT)   & True  & 3                        & 10                        \\
    Exp.\ 4 (MDP vs.\ TFT)     & False & 3                        & 5                         \\
    \bottomrule
  \end{tabular}
\end{table}

The POMDP formulation favoured longer planning horizons in self-play, likely because the noisy observation model introduces additional uncertainty that benefits from deeper policy tree search. In the cross-play setting, shorter update intervals (5 to 10 steps) were preferred over the self-play optimum (50 steps), suggesting faster belief revision is advantageous when facing stationary strategies like Tit-for-Tat.

\subsubsection*{A.4 Main Experiment Environment Parameters}

\begin{table}[h]
  \centering
  \caption{Fixed environment parameters used in all main experiments.}
  \label{tab:env-params}
  \begin{tabular}{lll}
    \toprule
    Parameter                   & Value                    & Description                           \\
    \midrule
    \texttt{noise\_levels}      & 0.00 to 0.30 (step 0.05) & Environmental action-flip probability \\
    \texttt{repetitions}        & 30                       & Independent repetitions per condition \\
    \texttt{turns}              & 5,000                    & Game length per repetition            \\
    Game payoffs ($T, R, P, S$) & (5, 3, 1, 0)             & Standard Prisoner's Dilemma matrix    \\
    \bottomrule
  \end{tabular}
\end{table}

The number of repetitions was increased from 5 to 30 to reduce variance in the reported statistics.

\subsubsection*{A.5 Full Grid Results}

Tables~\ref{tab:grid-score-self}--\ref{tab:grid-cc-tft} detail performance across the full hyperparameter grid, reporting the mean per-turn score and mutual cooperation rate for all combinations of planning horizon $h \in \{1,3,5,10\}$ and B-matrix update interval $\Delta t \in \{5,10,20,50\}$ for both formulations, in self-play and against Tit-for-Tat. Each cell averages 5 repetitions of 1{,}000 turns. Bold rows indicate configurations selected for the main experiments (re-run using 30 repetitions of 5{,}000 turns).

Three key trends emerge. First, the update interval dictates self-play cooperation: slower B-matrix updates ($\Delta t = 50$) sustain cooperation at substantially higher noise levels than fast updates ($\Delta t = 5$), as frequent re-estimation amplifies noise-induced fluctuations in the cooperative prior. Second, the planning horizon is critical against Tit-for-Tat. A myopic planner ($h = 1$) collapses entirely because it cannot value reciprocation, whereas $h \geq 3$ recovers near-optimal cooperation. In self-play, however, the benefit of $h > 3$ is marginal. Third, at low-to-moderate noise with matched hyperparameters (e.g., $h = 3$, $\Delta t = 50$, $\varepsilon = 0.10$), the POMDP sustains higher cooperation than the MDP ($0.74$ vs. $0.57$). This advantage reverses beyond the self-play threshold once the POMDP collapses (\cref{sec:coop-reward}).

\begin{table}[htbp]
  \centering
  \small
  \setlength{\tabcolsep}{4pt}
  \caption{Mean per-turn score across the full hyperparameter grid, self-play. Per-turn score lies between the mutual-defection payoff ($1$) and the mutual-cooperation payoff ($3$). Bold: configuration selected for the main experiments.}
  \label{tab:grid-score-self}
  \begin{tabular}{cc rrrrrrrr}
    \toprule
    $h$ & $\Delta t$ & \multicolumn{8}{c}{Noise level $\varepsilon$} \\
    \cmidrule(lr){3-10}
        &            & 0.00 & 0.05 & 0.10 & 0.15 & 0.20 & 0.25 & 0.30 & 0.35 \\
    \midrule
    \multicolumn{10}{l}{\textit{MDP}} \\
    1 & 5 & 2.98 & 2.09 & 1.63 & 1.50 & 1.63 & 1.74 & 1.89 & 1.97 \\
    1 & 10 & 2.97 & 2.33 & 2.11 & 1.67 & 1.72 & 1.71 & 1.82 & 1.93 \\
    1 & 20 & 2.94 & 2.69 & 2.51 & 1.80 & 1.83 & 1.79 & 1.89 & 1.95 \\
    1 & 50 & 2.85 & 2.65 & 2.67 & 2.55 & 2.18 & 1.85 & 1.86 & 1.95 \\
    3 & 5 & 2.97 & 1.87 & 1.58 & 1.68 & 1.65 & 1.82 & 1.95 & 1.98 \\
    3 & 10 & 2.96 & 2.12 & 1.84 & 1.72 & 1.71 & 1.89 & 1.91 & 1.98 \\
    3 & 20 & 2.94 & 2.47 & 2.24 & 2.21 & 2.04 & 1.89 & 1.97 & 2.06 \\
    \textbf{3} & \textbf{50} & \textbf{2.85} & \textbf{2.66} & \textbf{2.61} & \textbf{2.51} & \textbf{2.28} & \textbf{1.93} & \textbf{1.98} & \textbf{2.02} \\
    5 & 5 & 2.97 & 1.78 & 1.53 & 1.53 & 1.67 & 1.85 & 1.90 & 2.01 \\
    5 & 10 & 2.96 & 2.07 & 1.93 & 1.67 & 1.79 & 1.86 & 1.92 & 1.99 \\
    5 & 20 & 2.94 & 2.46 & 2.33 & 2.27 & 2.07 & 1.88 & 1.94 & 2.06 \\
    5 & 50 & 2.85 & 2.63 & 2.61 & 2.53 & 2.26 & 1.92 & 1.95 & 2.02 \\
    10 & 5 & 2.97 & 1.51 & 1.48 & 1.54 & 1.67 & 1.85 & 1.90 & 2.01 \\
    10 & 10 & 2.96 & 1.86 & 1.87 & 1.69 & 1.79 & 1.87 & 1.92 & 1.99 \\
    10 & 20 & 2.94 & 2.51 & 2.34 & 2.27 & 2.07 & 1.88 & 1.94 & 2.04 \\
    10 & 50 & 2.85 & 2.61 & 2.61 & 2.53 & 2.26 & 1.92 & 1.95 & 2.02 \\
    \midrule
    \multicolumn{10}{l}{\textit{POMDP}} \\
    1 & 5 & 2.98 & 1.67 & 1.68 & 1.50 & 1.73 & 1.72 & 1.84 & 1.95 \\
    1 & 10 & 2.97 & 2.16 & 1.77 & 1.71 & 1.73 & 1.74 & 1.85 & 1.95 \\
    1 & 20 & 2.94 & 2.68 & 2.11 & 1.87 & 1.91 & 1.76 & 1.85 & 1.96 \\
    1 & 50 & 2.85 & 2.80 & 2.65 & 2.28 & 2.16 & 1.80 & 1.88 & 1.96 \\
    3 & 5 & 2.97 & 1.82 & 1.79 & 1.83 & 1.79 & 1.83 & 1.93 & 1.98 \\
    3 & 10 & 2.96 & 2.29 & 1.92 & 1.85 & 1.81 & 1.86 & 1.93 & 1.98 \\
    3 & 20 & 2.94 & 2.49 & 2.39 & 2.00 & 1.87 & 1.87 & 1.93 & 1.99 \\
    3 & 50 & 2.85 & 2.80 & 2.76 & 2.32 & 1.89 & 1.94 & 1.92 & 1.99 \\
    5 & 5 & 2.97 & 1.83 & 1.93 & 1.78 & 1.77 & 1.82 & 1.90 & 1.98 \\
    5 & 10 & 2.96 & 2.33 & 2.03 & 1.81 & 1.87 & 1.87 & 1.92 & 1.98 \\
    5 & 20 & 2.94 & 2.50 & 2.38 & 1.99 & 1.87 & 1.87 & 1.93 & 1.99 \\
    \textbf{5} & \textbf{50} & \textbf{2.85} & \textbf{2.80} & \textbf{2.76} & \textbf{2.36} & \textbf{1.89} & \textbf{1.93} & \textbf{1.93} & \textbf{1.99} \\
    10 & 5 & 2.97 & 1.89 & 1.90 & 1.83 & 1.78 & 1.82 & 1.90 & 1.98 \\
    10 & 10 & 2.96 & 2.43 & 1.98 & 1.82 & 1.87 & 1.87 & 1.92 & 1.98 \\
    10 & 20 & 2.94 & 2.49 & 2.39 & 2.00 & 1.87 & 1.87 & 1.93 & 1.99 \\
    10 & 50 & 2.85 & 2.80 & 2.76 & 2.35 & 1.89 & 1.93 & 1.93 & 1.99 \\
    \bottomrule
  \end{tabular}
\end{table}

\begin{table}[htbp]
  \centering
  \small
  \setlength{\tabcolsep}{4pt}
  \caption{Mean per-turn score across the full hyperparameter grid, vs Tit-for-Tat. Per-turn score lies between the mutual-defection payoff ($1$) and the mutual-cooperation payoff ($3$). Bold: configuration selected for the main experiments.}
  \label{tab:grid-score-tft}
  \begin{tabular}{cc rrrrrrrr}
    \toprule
    $h$ & $\Delta t$ & \multicolumn{8}{c}{Noise level $\varepsilon$} \\
    \cmidrule(lr){3-10}
        &            & 0.00 & 0.05 & 0.10 & 0.15 & 0.20 & 0.25 & 0.30 & 0.35 \\
    \midrule
    \multicolumn{10}{l}{\textit{MDP}} \\
    1 & 5 & 1.04 & 1.36 & 1.61 & 1.83 & 2.03 & 2.15 & 2.24 & 2.33 \\
    1 & 10 & 1.04 & 1.36 & 1.61 & 1.83 & 2.05 & 2.16 & 2.25 & 2.32 \\
    1 & 20 & 1.05 & 1.38 & 1.61 & 1.85 & 2.05 & 2.15 & 2.25 & 2.33 \\
    1 & 50 & 1.13 & 1.41 & 1.65 & 1.86 & 2.06 & 2.16 & 2.23 & 2.32 \\
    \textbf{3} & \textbf{5} & \textbf{2.97} & \textbf{2.77} & \textbf{2.57} & \textbf{2.37} & \textbf{2.25} & \textbf{2.26} & \textbf{2.28} & \textbf{2.27} \\
    3 & 10 & 2.96 & 2.77 & 2.61 & 2.31 & 2.29 & 2.23 & 2.27 & 2.32 \\
    3 & 20 & 2.93 & 2.75 & 2.56 & 2.37 & 2.25 & 2.21 & 2.26 & 2.33 \\
    3 & 50 & 2.83 & 2.70 & 2.55 & 2.33 & 2.23 & 2.26 & 2.24 & 2.32 \\
    5 & 5 & 2.96 & 2.77 & 2.55 & 2.36 & 2.19 & 2.25 & 2.25 & 2.29 \\
    5 & 10 & 2.95 & 2.77 & 2.57 & 2.35 & 2.15 & 2.24 & 2.26 & 2.29 \\
    5 & 20 & 2.92 & 2.75 & 2.55 & 2.35 & 2.21 & 2.24 & 2.26 & 2.33 \\
    5 & 50 & 2.83 & 2.70 & 2.56 & 2.34 & 2.26 & 2.23 & 2.26 & 2.32 \\
    10 & 5 & 2.96 & 2.77 & 2.56 & 2.36 & 2.17 & 2.25 & 2.25 & 2.29 \\
    10 & 10 & 2.95 & 2.77 & 2.59 & 2.35 & 2.14 & 2.23 & 2.25 & 2.29 \\
    10 & 20 & 2.92 & 2.75 & 2.55 & 2.37 & 2.21 & 2.24 & 2.26 & 2.31 \\
    10 & 50 & 2.83 & 2.70 & 2.56 & 2.34 & 2.26 & 2.23 & 2.26 & 2.32 \\
    \midrule
    \multicolumn{10}{l}{\textit{POMDP}} \\
    1 & 5 & 1.04 & 1.36 & 1.62 & 1.83 & 2.04 & 2.16 & 2.24 & 2.32 \\
    1 & 10 & 1.04 & 1.36 & 1.62 & 1.85 & 2.05 & 2.15 & 2.26 & 2.32 \\
    1 & 20 & 1.05 & 1.39 & 1.63 & 1.86 & 2.06 & 2.16 & 2.25 & 2.32 \\
    1 & 50 & 1.13 & 1.40 & 1.65 & 1.85 & 2.05 & 2.16 & 2.24 & 2.32 \\
    3 & 5 & 2.97 & 2.78 & 2.59 & 2.30 & 2.28 & 2.26 & 2.24 & 2.33 \\
    \textbf{3} & \textbf{10} & \textbf{2.96} & \textbf{2.78} & \textbf{2.57} & \textbf{2.37} & \textbf{2.28} & \textbf{2.22} & \textbf{2.26} & \textbf{2.32} \\
    3 & 20 & 2.93 & 2.75 & 2.55 & 2.32 & 2.22 & 2.20 & 2.25 & 2.32 \\
    3 & 50 & 2.83 & 2.70 & 2.53 & 2.26 & 2.19 & 2.19 & 2.24 & 2.31 \\
    5 & 5 & 2.96 & 2.78 & 2.60 & 2.36 & 2.21 & 2.23 & 2.26 & 2.33 \\
    5 & 10 & 2.95 & 2.77 & 2.57 & 2.37 & 2.24 & 2.24 & 2.26 & 2.30 \\
    5 & 20 & 2.92 & 2.75 & 2.59 & 2.34 & 2.23 & 2.21 & 2.25 & 2.32 \\
    5 & 50 & 2.83 & 2.70 & 2.52 & 2.28 & 2.19 & 2.18 & 2.24 & 2.31 \\
    10 & 5 & 2.96 & 2.78 & 2.58 & 2.36 & 2.22 & 2.23 & 2.26 & 2.33 \\
    10 & 10 & 2.95 & 2.77 & 2.57 & 2.36 & 2.24 & 2.25 & 2.27 & 2.30 \\
    10 & 20 & 2.92 & 2.75 & 2.60 & 2.35 & 2.23 & 2.22 & 2.25 & 2.32 \\
    10 & 50 & 2.83 & 2.70 & 2.52 & 2.28 & 2.19 & 2.18 & 2.24 & 2.31 \\
    \bottomrule
  \end{tabular}
\end{table}

\begin{table}[htbp]
  \centering
  \small
  \setlength{\tabcolsep}{4pt}
  \caption{Mutual cooperation rate across the full hyperparameter grid, self-play. Bold: configuration selected for the main experiments.}
  \label{tab:grid-cc-self}
  \begin{tabular}{cc rrrrrrrr}
    \toprule
    $h$ & $\Delta t$ & \multicolumn{8}{c}{Noise level $\varepsilon$} \\
    \cmidrule(lr){3-10}
        &            & 0.00 & 0.05 & 0.10 & 0.15 & 0.20 & 0.25 & 0.30 & 0.35 \\
    \midrule
    \multicolumn{10}{l}{\textit{MDP}} \\
    1 & 5 & 0.99 & 0.42 & 0.15 & 0.03 & 0.05 & 0.07 & 0.12 & 0.13 \\
    1 & 10 & 0.98 & 0.52 & 0.33 & 0.10 & 0.09 & 0.07 & 0.09 & 0.12 \\
    1 & 20 & 0.97 & 0.70 & 0.52 & 0.15 & 0.13 & 0.10 & 0.12 & 0.12 \\
    1 & 50 & 0.92 & 0.69 & 0.65 & 0.52 & 0.29 & 0.12 & 0.10 & 0.13 \\
    3 & 5 & 0.99 & 0.17 & 0.08 & 0.07 & 0.06 & 0.10 & 0.14 & 0.13 \\
    3 & 10 & 0.98 & 0.32 & 0.19 & 0.09 & 0.07 & 0.12 & 0.13 & 0.14 \\
    3 & 20 & 0.97 & 0.57 & 0.36 & 0.27 & 0.18 & 0.13 & 0.15 & 0.17 \\
    \textbf{3} & \textbf{50} & \textbf{0.92} & \textbf{0.64} & \textbf{0.57} & \textbf{0.46} & \textbf{0.33} & \textbf{0.15} & \textbf{0.15} & \textbf{0.15} \\
    5 & 5 & 0.98 & 0.16 & 0.04 & 0.03 & 0.06 & 0.10 & 0.12 & 0.14 \\
    5 & 10 & 0.98 & 0.31 & 0.20 & 0.06 & 0.10 & 0.11 & 0.13 & 0.14 \\
    5 & 20 & 0.97 & 0.52 & 0.35 & 0.28 & 0.20 & 0.13 & 0.13 & 0.17 \\
    5 & 50 & 0.92 & 0.62 & 0.58 & 0.45 & 0.32 & 0.14 & 0.14 & 0.16 \\
    10 & 5 & 0.98 & 0.03 & 0.03 & 0.03 & 0.06 & 0.10 & 0.12 & 0.14 \\
    10 & 10 & 0.98 & 0.20 & 0.20 & 0.07 & 0.10 & 0.11 & 0.13 & 0.14 \\
    10 & 20 & 0.97 & 0.56 & 0.34 & 0.28 & 0.20 & 0.13 & 0.13 & 0.16 \\
    10 & 50 & 0.92 & 0.59 & 0.58 & 0.46 & 0.32 & 0.14 & 0.14 & 0.16 \\
    \midrule
    \multicolumn{10}{l}{\textit{POMDP}} \\
    1 & 5 & 0.99 & 0.15 & 0.15 & 0.03 & 0.09 & 0.07 & 0.10 & 0.13 \\
    1 & 10 & 0.98 & 0.40 & 0.15 & 0.08 & 0.10 & 0.08 & 0.10 & 0.13 \\
    1 & 20 & 0.97 & 0.69 & 0.30 & 0.13 & 0.16 & 0.08 & 0.11 & 0.13 \\
    1 & 50 & 0.92 & 0.82 & 0.66 & 0.38 & 0.31 & 0.09 & 0.11 & 0.13 \\
    3 & 5 & 0.99 & 0.17 & 0.11 & 0.13 & 0.09 & 0.10 & 0.12 & 0.14 \\
    3 & 10 & 0.98 & 0.40 & 0.17 & 0.11 & 0.09 & 0.11 & 0.12 & 0.14 \\
    3 & 20 & 0.97 & 0.56 & 0.37 & 0.19 & 0.12 & 0.11 & 0.13 & 0.14 \\
    3 & 50 & 0.92 & 0.82 & 0.74 & 0.35 & 0.13 & 0.15 & 0.13 & 0.15 \\
    5 & 5 & 0.98 & 0.16 & 0.17 & 0.11 & 0.08 & 0.10 & 0.12 & 0.14 \\
    5 & 10 & 0.98 & 0.43 & 0.24 & 0.09 & 0.12 & 0.11 & 0.12 & 0.14 \\
    5 & 20 & 0.97 & 0.54 & 0.37 & 0.19 & 0.12 & 0.11 & 0.13 & 0.14 \\
    \textbf{5} & \textbf{50} & \textbf{0.92} & \textbf{0.82} & \textbf{0.74} & \textbf{0.37} & \textbf{0.13} & \textbf{0.14} & \textbf{0.13} & \textbf{0.15} \\
    10 & 5 & 0.98 & 0.18 & 0.17 & 0.14 & 0.08 & 0.10 & 0.12 & 0.14 \\
    10 & 10 & 0.98 & 0.50 & 0.21 & 0.09 & 0.12 & 0.11 & 0.12 & 0.14 \\
    10 & 20 & 0.97 & 0.54 & 0.37 & 0.19 & 0.12 & 0.11 & 0.13 & 0.14 \\
    10 & 50 & 0.92 & 0.82 & 0.74 & 0.37 & 0.13 & 0.14 & 0.13 & 0.15 \\
    \bottomrule
  \end{tabular}
\end{table}

\begin{table}[htbp]
  \centering
  \small
  \setlength{\tabcolsep}{4pt}
  \caption{Mutual cooperation rate across the full hyperparameter grid, vs Tit-for-Tat. Bold: configuration selected for the main experiments.}
  \label{tab:grid-cc-tft}
  \begin{tabular}{cc rrrrrrrr}
    \toprule
    $h$ & $\Delta t$ & \multicolumn{8}{c}{Noise level $\varepsilon$} \\
    \cmidrule(lr){3-10}
        &            & 0.00 & 0.05 & 0.10 & 0.15 & 0.20 & 0.25 & 0.30 & 0.35 \\
    \midrule
    \multicolumn{10}{l}{\textit{MDP}} \\
    1 & 5 & 0.00 & 0.01 & 0.02 & 0.04 & 0.06 & 0.09 & 0.13 & 0.16 \\
    1 & 10 & 0.00 & 0.01 & 0.02 & 0.04 & 0.07 & 0.10 & 0.13 & 0.16 \\
    1 & 20 & 0.00 & 0.01 & 0.02 & 0.04 & 0.06 & 0.10 & 0.14 & 0.16 \\
    1 & 50 & 0.00 & 0.01 & 0.02 & 0.05 & 0.07 & 0.10 & 0.13 & 0.16 \\
    \textbf{3} & \textbf{5} & \textbf{0.96} & \textbf{0.80} & \textbf{0.57} & \textbf{0.38} & \textbf{0.21} & \textbf{0.20} & \textbf{0.19} & \textbf{0.19} \\
    3 & 10 & 0.95 & 0.81 & 0.67 & 0.27 & 0.23 & 0.15 & 0.16 & 0.17 \\
    3 & 20 & 0.94 & 0.79 & 0.58 & 0.36 & 0.20 & 0.16 & 0.16 & 0.17 \\
    3 & 50 & 0.86 & 0.76 & 0.58 & 0.29 & 0.19 & 0.19 & 0.18 & 0.17 \\
    5 & 5 & 0.95 & 0.80 & 0.49 & 0.30 & 0.15 & 0.16 & 0.17 & 0.18 \\
    5 & 10 & 0.95 & 0.81 & 0.58 & 0.33 & 0.16 & 0.16 & 0.15 & 0.19 \\
    5 & 20 & 0.93 & 0.79 & 0.54 & 0.29 & 0.19 & 0.20 & 0.16 & 0.17 \\
    5 & 50 & 0.86 & 0.76 & 0.62 & 0.28 & 0.20 & 0.18 & 0.19 & 0.17 \\
    10 & 5 & 0.95 & 0.80 & 0.51 & 0.33 & 0.15 & 0.16 & 0.17 & 0.18 \\
    10 & 10 & 0.95 & 0.80 & 0.63 & 0.33 & 0.16 & 0.15 & 0.15 & 0.19 \\
    10 & 20 & 0.93 & 0.79 & 0.55 & 0.29 & 0.19 & 0.19 & 0.16 & 0.18 \\
    10 & 50 & 0.86 & 0.76 & 0.62 & 0.29 & 0.20 & 0.18 & 0.19 & 0.17 \\
    \midrule
    \multicolumn{10}{l}{\textit{POMDP}} \\
    1 & 5 & 0.00 & 0.01 & 0.02 & 0.04 & 0.07 & 0.10 & 0.13 & 0.16 \\
    1 & 10 & 0.00 & 0.01 & 0.02 & 0.05 & 0.07 & 0.10 & 0.13 & 0.16 \\
    1 & 20 & 0.00 & 0.01 & 0.03 & 0.05 & 0.07 & 0.10 & 0.13 & 0.17 \\
    1 & 50 & 0.00 & 0.02 & 0.03 & 0.05 & 0.07 & 0.10 & 0.13 & 0.16 \\
    3 & 5 & 0.96 & 0.82 & 0.63 & 0.26 & 0.30 & 0.17 & 0.16 & 0.17 \\
    \textbf{3} & \textbf{10} & \textbf{0.95} & \textbf{0.82} & \textbf{0.63} & \textbf{0.38} & \textbf{0.24} & \textbf{0.15} & \textbf{0.16} & \textbf{0.18} \\
    3 & 20 & 0.94 & 0.80 & 0.58 & 0.27 & 0.16 & 0.15 & 0.15 & 0.17 \\
    3 & 50 & 0.86 & 0.76 & 0.57 & 0.20 & 0.15 & 0.14 & 0.15 & 0.18 \\
    5 & 5 & 0.95 & 0.82 & 0.65 & 0.33 & 0.24 & 0.16 & 0.18 & 0.18 \\
    5 & 10 & 0.95 & 0.81 & 0.62 & 0.40 & 0.25 & 0.16 & 0.16 & 0.17 \\
    5 & 20 & 0.93 & 0.80 & 0.64 & 0.32 & 0.18 & 0.16 & 0.15 & 0.17 \\
    5 & 50 & 0.86 & 0.75 & 0.52 & 0.23 & 0.15 & 0.13 & 0.15 & 0.18 \\
    10 & 5 & 0.95 & 0.82 & 0.62 & 0.33 & 0.24 & 0.16 & 0.18 & 0.18 \\
    10 & 10 & 0.95 & 0.82 & 0.62 & 0.36 & 0.25 & 0.16 & 0.16 & 0.17 \\
    10 & 20 & 0.93 & 0.80 & 0.66 & 0.32 & 0.18 & 0.16 & 0.15 & 0.17 \\
    10 & 50 & 0.86 & 0.75 & 0.52 & 0.23 & 0.15 & 0.13 & 0.15 & 0.18 \\
    \bottomrule
  \end{tabular}
\end{table}

\section{Results Without Epistemic Contribution}
\label{app:no-epistemic}

\begin{figure}[h]
  \centering
  \includegraphics[width=\linewidth]{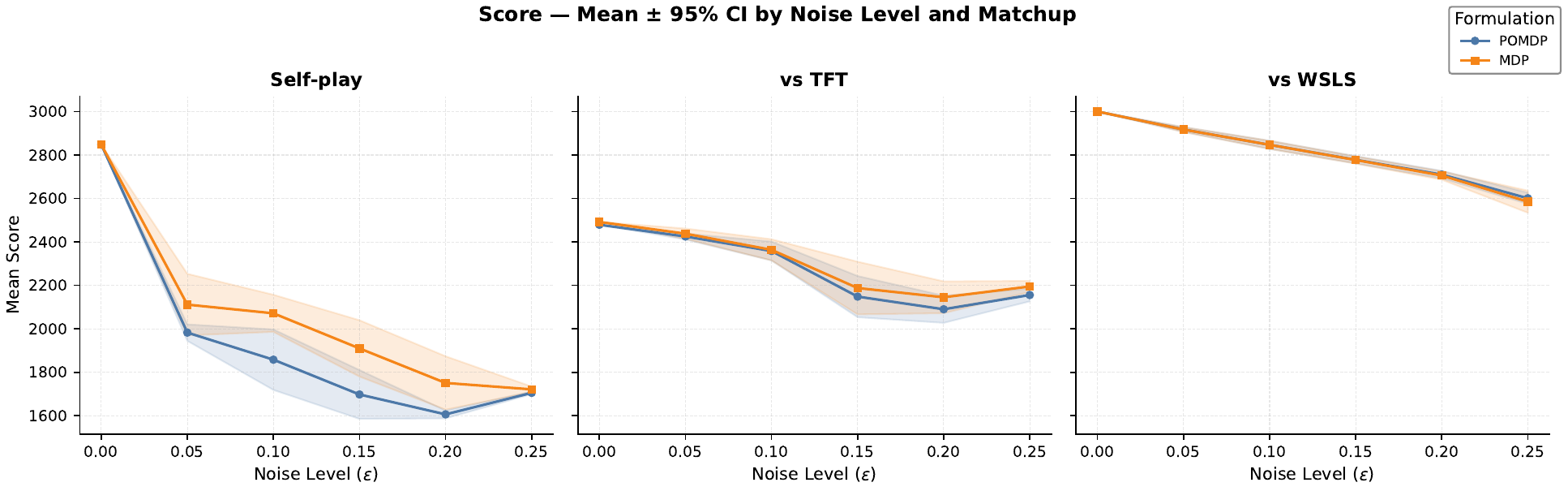}
  \caption{Mean cumulative score ($\pm$ 95\% CI) as a function of noise level $\epsilon$, with the epistemic (parameter/state information gain) term removed from the EFE.}
  \label{fig:pragmatic-scores}
\end{figure}

\begin{figure}[h]
  \centering
  \includegraphics[width=\linewidth]{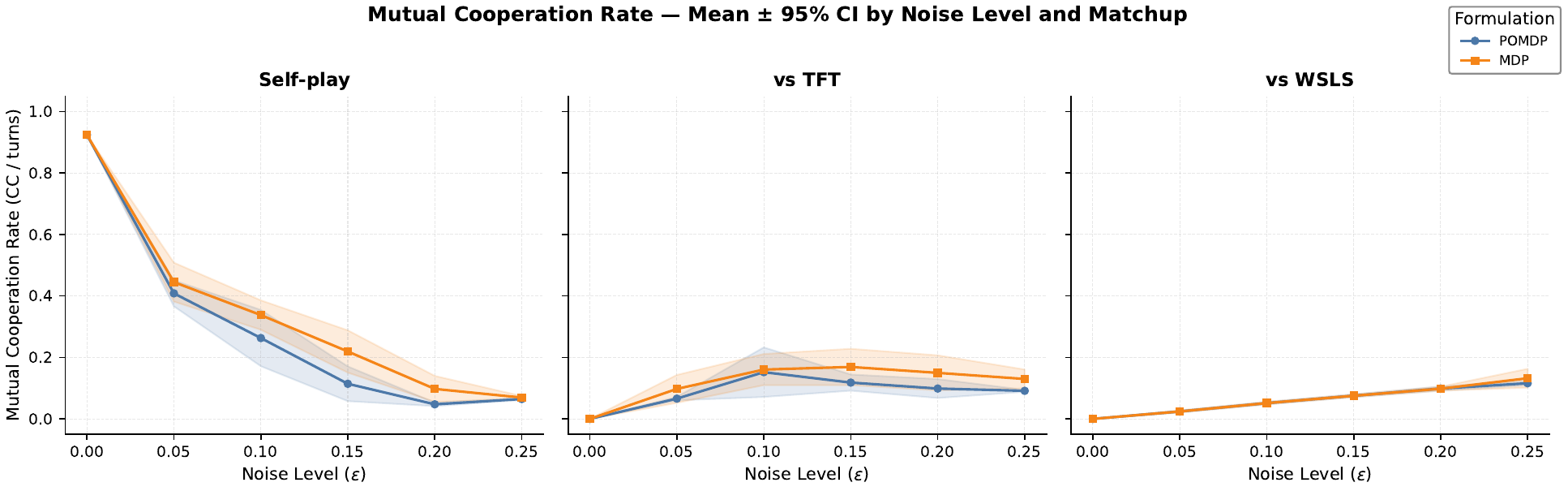}
  \caption{Mutual cooperation rate as a function of $\epsilon$, with the epistemic term removed from the EFE.}
  \label{fig:pragmatic-coop}
\end{figure}

To isolate the contribution of epistemic drive, we evaluate both formulations using only the pragmatic value component of the expected free energy, zeroing out both state and parameter information gain terms. \Cref{fig:pragmatic-scores,fig:pragmatic-coop} show score and cooperation rate under this ablation.

Without epistemic drive, both formulations fail to sustain cooperation at any non-trivial noise level. Against TFT, neither agent exceeds 0.2 mutual cooperation, compared to the full-EFE POMDP's 0.70 at $\epsilon = 0.10$, confirming that the cooperative advantages reported in \cref{sec:efe-decomp} are not a product of the POMDP representation alone but depend on parameter information gain driving B-matrix learning. Crucially, the POMDP's performance is not merely reduced but is marginally worse than the MDP's across conditions: without the learning signal needed to discover the opponent's reactive policy, the cooperative prior becomes an uncompensated bias.

The MDP with full EFE (\cref{sec:efe-decomp}) shows that epistemic drive under an inadequate representation yields only gradual, uncommitted improvement. The POMDP with pragmatic value only shows that the correct representation without epistemic drive yields no improvement at all. The cooperative advantages we report require both: the POMDP provides the latent structure over which epistemic value can operate, and epistemic value drives the active information-seeking needed to learn the opponent's policy through noise.

\section{Results in the Stag Hunt}
\label{app:stag-hunt}

\begin{figure}[h]
  \centering
  \includegraphics[width=\linewidth]{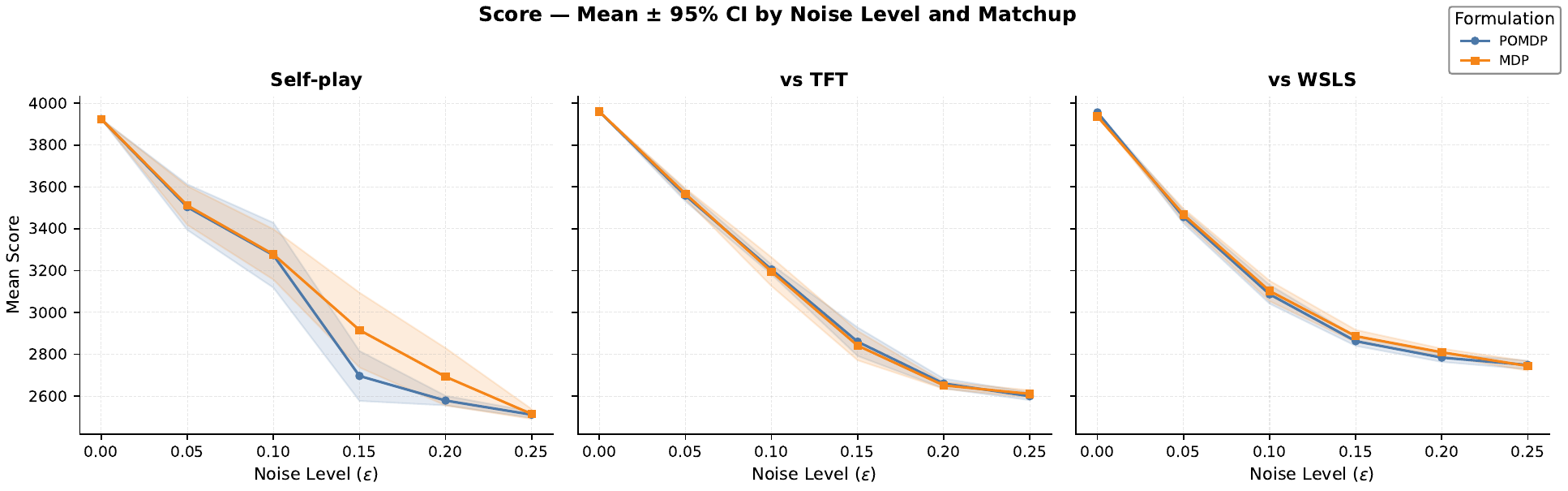}
  \caption{Mean cumulative score ($\pm$ 95\% CI) as a function of noise level $\epsilon$ in the Stag Hunt game.}
  \label{fig:stag-scores}
\end{figure}

\begin{figure}[h]
  \centering
  \includegraphics[width=\linewidth]{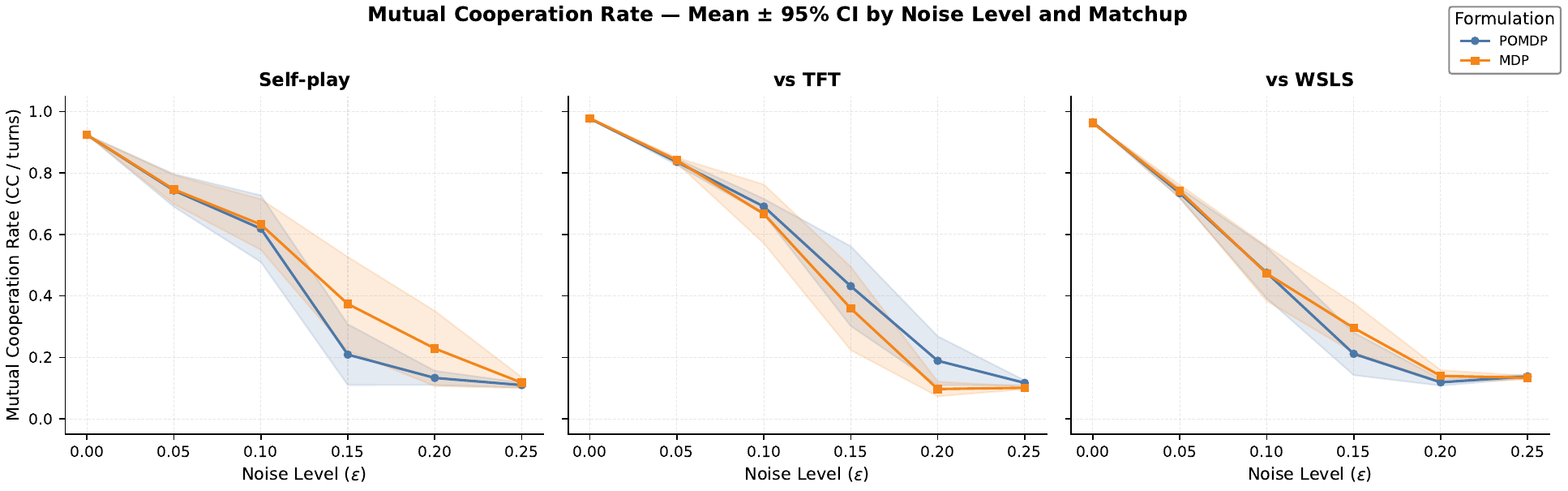}
  \caption{Mutual cooperation rate as a function of $\epsilon$ in the Stag Hunt game.}
  \label{fig:stag-coop}
\end{figure}

The Stag Hunt provides a complementary test environment with a qualitatively different strategic structure. Two players simultaneously choose to hunt Stag (cooperate) or Hare (defect). Hunting Stag succeeds only if both players coordinate, yielding the highest payoff to each ($R=4$), while hunting Hare provides a safe but lower payoff regardless of the partner's choice ($P=2$ if both hunt Hare, $T=3$ if the partner hunts Stag). Failing to coordinate on Stag yields the worst outcome for the Stag hunter ($S=1$). We use the standard payoffs $S=1$, $P=2$, $T=3$, $R=4$ with the same binary symmetric noise channel as in the IPD experiments.

The crucial difference from the IPD is the absence of exploitation incentive: mutual cooperation is both the socially optimal and individually preferred outcome ($R > T$), so there is no temptation to defect against a cooperating partner. The strategic challenge is therefore pure coordination under risk, where each player must be confident enough in the other's cooperative intent to justify the vulnerability of choosing Stag, rather than the tension between cooperation and exploitation that defines the IPD.

\Cref{fig:stag-scores,fig:stag-coop} show the mean cumulative score and mutual cooperation rate across noise levels. The POMDP advantage that was prominent in the IPD largely disappears. In self-play, the two formulations track closely, with the POMDP marginally ahead at $\epsilon = 0.05$ to $0.10$ but within overlapping confidence intervals. Against TFT, the MDP slightly outperforms at low noise, with the two converging at higher noise. Against WSLS, performance is essentially identical. Neither formulation exhibits the sharp collapse observed in IPD self-play; both degrade gracefully as noise increases.

These results clarify the scope of our contribution. The POMDP formulation's advantage is specific to games where the payoff structure makes intent attribution decision-relevant. In the IPD, observing defection demands a judgement about whether the cause was exploitation or error, and the answer determines whether retaliation or forgiveness is optimal. In the Stag Hunt, the absence of a temptation payoff ($T < R$) means this distinction is less consequential: even misattributed intent does not alter the optimal response as dramatically, because there is no exploitative counter-strategy to guard against. The MDP's conflation of aleatoric and epistemic uncertainty, which is costly in the IPD, carries a lighter penalty when the game's incentives are aligned.

\section{Critical Noise Threshold: Extended Analysis}
\label{app:self-play-collapse}

This appendix derives the myopic attribution boundary as the foundation for \cref{prop:hstep}, proves the proposition, and traces the mechanisms underlying the empirical phase transition reported in \cref{sec:attribution-verification}.

\subsection*{Myopic Attribution Boundary}
\label{app:myopic-boundary}

The simplest case of \cref{prop:hstep} evaluates a single observation ($h = 1$). When the agent observes opponent defection, the posterior probability of cooperative intent follows from the binary symmetric channel:
\begin{equation}
  P(\text{int.}\ C \mid \text{obs.}\ D) = \frac{\varepsilon\, p_t(\text{CC})}{\varepsilon\, p_t(\text{CC}) + (1-\varepsilon)(1-p_t(\text{CC}))}
  \label{eq:attribution}
\end{equation}
Setting this equal to $0.5$ and solving yields the myopic forgiveness condition:
\begin{equation}
  \varepsilon^* = 1 - p_t(\text{CC})
  \label{eq:threshold-myopic}
\end{equation}
When $p_t(\text{CC}) > 1-\varepsilon$, observed defection is more likely noise than hostile intent and the agent forgives. When $p_t(\text{CC}) < 1-\varepsilon$, the agent attributes hostile intent and retaliates.

\subsection*{System-Level Fixed Point}
\label{app:system-threshold}

Since $p_t(\text{CC})$ is learned from noisy observations via Dirichlet accumulation, higher noise degrades the achievable cooperative prior, creating a fixed-point condition. Let $p_\infty(\varepsilon)$ denote the converged cooperative prior under noise level $\varepsilon$. Cooperation is sustained when $p_\infty(\varepsilon) > 1 - \varepsilon$; the critical threshold $\tilde{\varepsilon}$ satisfies:
\begin{equation}
  \tilde{\varepsilon} = 1 - p_\infty(\tilde{\varepsilon})
  \label{eq:fixed-point}
\end{equation}
\Cref{tab:state-priors} confirms that $\tilde{\varepsilon}$ lies near $0.10$ for the POMDP in self-play: CC satisfies the myopic forgiveness condition at $\varepsilon = 0.05$ and $\varepsilon = 0.10$ but not at $\varepsilon = 0.15$, and the posterior $P(\text{int.}\ C \mid \text{obs.}\ D)$ sits just above $0.5$ at both forgiving noise levels.
\begin{table}[h]
  \centering
  \caption{POMDP cooperative priors across all states (self-play, $t = 500$). Only CC satisfies the myopic forgiveness condition at any noise level tested.}
  \label{tab:state-priors}
  \begin{tabular}{lclclcl}
    \toprule
    State & $p_{500}(0.05)$ & Forgiving?                 & $p_{500}(0.10)$ & Forgiving?                 & $p_{500}(0.15)$ & Forgiving?                 \\
    \midrule
    CC    & 0.9541          & \checkmark                 & 0.9070          & \checkmark                 & 0.7612          & $\times$                   \\
    CD    & 0.8926          & $\times$ (need $p > 0.95$) & 0.8662          & $\times$ (need $p > 0.90$) & 0.6779          & $\times$ (need $p > 0.85$) \\
    DC    & 0.4508          & $\times$                   & 0.6186          & $\times$                   & 0.5597          & $\times$                   \\
    DD    & 0.8791          & $\times$ (need $p > 0.95$) & 0.8110          & $\times$ (need $p > 0.90$) & 0.6853          & $\times$ (need $p > 0.85$) \\
    \bottomrule
  \end{tabular}
\end{table}

\subsection*{Proof of \cref{prop:hstep}}

\begin{proof}
  The agent forgives $k$ defections in $h$ rounds when the posterior favours $H_C$:
  \begin{equation}
    \frac{
      p^h \cdot \varepsilon^k (1-\varepsilon)^{h-k}
    }{
      p^h \cdot \varepsilon^k (1-\varepsilon)^{h-k}
      \;+\;
      (1-p^h)\,(1-\varepsilon)^k\, \varepsilon^{h-k}
    }
    \;>\; \tfrac{1}{2}
  \end{equation}
  The first term of the denominator equals the numerator, so the posterior exceeds $\tfrac{1}{2}$ exactly when the numerator exceeds the second (hostile) term:
  \begin{equation}
    p^h \, \varepsilon^k (1-\varepsilon)^{h-k}
    \;>\;
    (1-p^h)\,(1-\varepsilon)^k\, \varepsilon^{h-k}.
  \end{equation}
  Taking logarithms and collecting the terms in $k$ gives
  \begin{equation}
    \ln\!\frac{p^h}{1-p^h}
    \;>\;
    \bigl(k - (h-k)\bigr)\,\ln\frac{1-\varepsilon}{\varepsilon}
    \;=\;
    (2k - h)\,\ln\frac{1-\varepsilon}{\varepsilon}.
  \end{equation}
  Since $\ln\bigl((1-\varepsilon)/\varepsilon\bigr) > 0$ for $\varepsilon < \tfrac{1}{2}$, dividing through and solving for $k$ yields the log-likelihood ratio condition
  \begin{equation}
    k \;<\; \frac{1}{2}\!\left(
    h \;+\;
    \frac{\ln\!\bigl(p^h/(1-p^h)\bigr)}{\ln\!\bigl((1-\varepsilon)/\varepsilon\bigr)}
    \right)
  \end{equation}
  and $k^*$ is the largest integer satisfying this inequality. At $h=1$, the condition reduces to $\ln\!\bigl(p/(1-p)\bigr) > 0$, which holds when $p > \tfrac{1}{2}$. Since the myopic condition $p > 1-\varepsilon$ implies $p > \tfrac{1}{2}$ for all $\varepsilon < 0.5$, this recovers \cref{eq:threshold-myopic}.
\end{proof}

\subsection*{Extension to Mixed-Intent Sequences}
\label{app:mixed-intent}

The binary test underlying \cref{prop:hstep} partitions the hypothesis space into $H_h$ (fully cooperative) and $H_0$ (fully hostile). The true complement of $H_h$ includes mixed sequences where the opponent cooperates for some rounds and defects for others. We show the binary threshold is an upper bound on the threshold under the full mixture, and that the gap is at most one defection in the paper's parameter range.

\begin{definition}[Cooperative-count hypotheses]
  For a planning window of $h$ rounds, let $H_\tau$ denote the hypothesis that the opponent intends $C$ for exactly $\tau$ of the $h$ rounds and $D$ for the remaining $h - \tau$ rounds, for $\tau \in \{0, 1, \ldots, h\}$. Since the noise channel is memoryless, the likelihood of the observed defection count $k$ depends only on $\tau$, not on which rounds are cooperative.
\end{definition}

Under $H_\tau$, the $\tau$ cooperative rounds each produce a defection with probability $\varepsilon$, and the $h - \tau$ hostile rounds each produce a defection with probability $1 - \varepsilon$. The total defection count is distributed as the sum of two independent binomials, giving:
\begin{equation}
  P(k \mid H_\tau) = \sum_{j=\max(0,\, k-(h-\tau))}^{\min(k,\, \tau)}
  \binom{\tau}{j} \varepsilon^j (1-\varepsilon)^{\tau - j}
  \;\cdot\;
  \binom{h-\tau}{k-j} (1-\varepsilon)^{k-j} \varepsilon^{h-\tau-(k-j)}
  \label{eq:mixed-likelihood}
\end{equation}
where $j$ indexes the number of noise-driven defections from the cooperative rounds. Under an i.i.d.\ prior with per-round cooperation probability $p = p_t(\text{CC})$, the number of cooperative rounds follows $\tau \sim \text{Bin}(h, p)$, so $P(H_\tau) = \binom{h}{\tau} p^\tau (1-p)^{h-\tau}$.

The agent forgives when the posterior favours $H_h$ over the aggregate of all non-fully-cooperative hypotheses:
\begin{equation}
  \frac{P(H_h) \cdot P(k \mid H_h)}{\sum_{\tau=0}^{h-1} P(H_\tau) \cdot P(k \mid H_\tau)} > 1
  \label{eq:mixed-forgiveness}
\end{equation}

Let $k^*_{\text{mixed}}$ denote the threshold under \cref{eq:mixed-forgiveness} and $k^*_{\text{binary}}$ the threshold from \cref{prop:hstep}.

\begin{proposition}[Upper bound from the binary test]
  \label{prop:mixed-bound}
  $k^*_{\emph{mixed}}(h, p, \varepsilon) \leq k^*_{\emph{binary}}(h, p, \varepsilon)$.
\end{proposition}

\begin{proof}
  The mixed-model denominator is $(1 - p^h)$ times a weighted average of $P(k \mid H_\tau)$ over $\tau \in \{0, \ldots, h-1\}$ with normalised weights $w_\tau = P(H_\tau)/(1 - p^h)$. The binary-test denominator is $(1 - p^h) \cdot P(k \mid H_0)$. Since $H_0$ is one term in the weighted average with weight $w_0 > 0$ and all remaining terms are non-negative:
  \[
    \sum_{\tau=0}^{h-1} w_\tau \cdot P(k \mid H_\tau)
    \;=\; w_0 \cdot P(k \mid H_0) + \sum_{\tau=1}^{h-1} w_\tau \cdot P(k \mid H_\tau)
    \;\geq\; w_0 \cdot P(k \mid H_0).
  \]
  For the forgiveness-relevant regime ($k$ small), hypotheses with $\tau \geq 1$ assign higher likelihood to small $k$ than $H_0$ does, since their expected defection count is strictly lower. We verify $P(k \mid H_\tau) \geq P(k \mid H_0)$ computationally for all parameter combinations in \cref{tab:mixed-gap}, yielding:
  \[
    \sum_{\tau=0}^{h-1} w_\tau \cdot P(k \mid H_\tau) \;\geq\; P(k \mid H_0).
  \]
  The mixed-model denominator is therefore at least as large as the binary denominator. Since the numerator $p^h \cdot P(k \mid H_h)$ is identical in both models, the mixed-model posterior ratio is weakly smaller for every $k$, and hence $k^*_{\text{mixed}} \leq k^*_{\text{binary}}$.
\end{proof}

\begin{table}[h]
  \centering
  \caption{Binary vs.\ mixed-model forgiveness thresholds at $h = 5$. $R(k)$: ratio of mixed-model denominator to binary denominator at the forgiveness boundary.}
  \label{tab:mixed-gap}
  \begin{tabular}{cccccc}
    \toprule
    $p$   & $\varepsilon$ & $k^*_{\text{binary}}$ & $R(k^*_{\text{binary}})$ & $k^*_{\text{mixed}}$ & Gap \\
    \midrule
    0.954 & 0.05          & 1                     & 1.01                     & 1                    & 0   \\
    0.954 & 0.10          & 1                     & 1.05                     & 1                    & 0   \\
    0.954 & 0.15          & 2                     & 1.12                     & 2                    & 0   \\
    0.907 & 0.05          & 1                     & 1.03                     & 1                    & 0   \\
    0.907 & 0.10          & 1                     & 1.10                     & 1                    & 0   \\
    0.907 & 0.15          & 1                     & 1.24                     & 1                    & 0   \\
    0.761 & 0.05          & 1                     & 1.08                     & 1                    & 0   \\
    0.761 & 0.10          & 1                     & 1.29                     & 0                    & 1   \\
    0.761 & 0.15          & 1                     & 1.65                     & 0                    & 1   \\
    \bottomrule
  \end{tabular}
\end{table}

The gap is at most 1 across all parameter combinations. An important caveat is that \cref{tab:mixed-gap} treats $p$ and $\varepsilon$ as independent, but in practice $p = p_\infty(\varepsilon)$ is a function of the noise level: $p = 0.954$ was learned at $\varepsilon = 0.05$, $p = 0.907$ at $\varepsilon = 0.10$, and $p = 0.761$ at $\varepsilon = 0.15$. The off-diagonal entries (e.g., $p = 0.761$ at $\varepsilon = 0.05$) are hypothetical combinations that do not arise in the experiments---a cooperative prior that low would not be learned at that noise level. The empirically realised conditions lie on the diagonal: $(0.954, 0.05)$, $(0.907, 0.10)$, and $(0.761, 0.15)$. On this diagonal, the gap is 0 for the two noise levels where cooperation is still functioning ($\varepsilon = 0.05$ and $\varepsilon = 0.10$), and 1 only at $\varepsilon = 0.15$ where cooperation has already collapsed and the forgiveness threshold is moot. The binary approximation is therefore exact in every regime where it is operationally relevant.

The off-diagonal entries remain informative as a robustness check: they confirm that even under hypothetical prior--noise mismatches, the gap never exceeds 1. The binary test is tightest when $p$ is large (the binomial prior concentrates mass near $\tau = h$, so $H_0$ dominates the non-cooperative aggregate) and when $\varepsilon$ is small (all non-cooperative hypotheses assign similar likelihoods to small $k$). The largest gaps occur at lower $p$ and higher $\varepsilon$---precisely near the phase transition---where the intermediate hypotheses receive meaningful prior weight and assign substantially higher likelihood to small $k$ than $H_0$ does.

The binary partition is not merely a convenient approximation but reflects the game's strategic structure. For two rational agents in the IPD, the payoff structure drives play away from the asymmetric states (CD, DC) and toward the symmetric equilibria: mutual cooperation (CC) or mutual defection (DD). The mixed hypotheses $H_\tau$ for intermediate $\tau$ correspond to trajectories that pass through these asymmetric states---precisely the trajectories that rational play discourages. This game-theoretic argument is complemented by the bistable equilibrium structure: converged dynamics concentrate opponent behaviour at sustained cooperation or sustained defection, with mixed sequences confined to transient episodes near the basin boundary. The extended analysis confirms that the binary approximation is tight precisely because the game structure it approximates is itself approximately binary.

\subsection*{State-Dependent Priors and the Belief-Driven Collapse}

Execution noise does not directly alter the intention state. When both agents intend to cooperate, the true hidden state remains CC regardless of the noisy observation because the A matrix absorbs noise at the likelihood level. The system exits CC only when an agent's posterior falls below the attribution boundary, causing it to genuinely change its intention to defect.

CC is the only forgiving state at every noise level. Under myopic attribution, once a belief-driven retaliation shifts the true state out of CC, any subsequently observed defection is attributed to hostile intent. The cascade proceeds as follows: the focal agent's belief tips below $0.5$, it genuinely intends to defect, the true state shifts to DC, and from DC the cooperative prior is far below the forgiveness threshold. The opponent observes this genuine defection, attributes hostile intent, and retaliates. Each step lands the system in a state further from the attribution boundary, with no myopic path back to CC. These retaliatory episodes inject defection-labelled transitions into the Dirichlet counts, actively eroding $p_t(\text{CC})$ and making the collapse progressively more permanent.

\subsection*{Competing Effects of the Planning Horizon}

\Cref{prop:hstep} reveals two effects of the planning horizon that work in opposite directions. Evidence aggregation is beneficial: $k^*$ is non-decreasing in $h$, so longer windows tolerate more individual noisy defections. Intent persistence is costly: $p_t(\text{CC})^h$ decays exponentially in $h$, so longer windows are less likely to contain sustained cooperative intent. \Cref{tab:hstep-sweep} traces both quantities across $h = 1$ to $10$.

\begin{table}[h]
  \centering
  \caption{Maximum forgivable defection count $k^*$ and cooperative survival probability $p_{500}(\text{CC})^h$ as a function of planning horizon $h$ (POMDP, self-play).}
  \label{tab:hstep-sweep}
  \begin{tabular}{r cc cc cc}
    \toprule
        & \multicolumn{2}{c}{$\varepsilon = 0.05$} & \multicolumn{2}{c}{$\varepsilon = 0.10$} & \multicolumn{2}{c}{$\varepsilon = 0.15$}                                     \\
    \cmidrule(lr){2-3} \cmidrule(lr){4-5} \cmidrule(lr){6-7}
    $h$ & $k^*$                                    & $p_{500}^h$                              & $k^*$                                    & $p_{500}^h$ & $k^*$ & $p_{500}^h$ \\
    \midrule
    1   & 1                                        & 0.954                                    & 1                                        & 0.907       & 0     & 0.761       \\
    2   & 1                                        & 0.910                                    & 1                                        & 0.823       & 1     & 0.579       \\
    3   & 1                                        & 0.868                                    & 1                                        & 0.746       & 1     & 0.441       \\
    4   & 2                                        & 0.828                                    & 2                                        & 0.677       & 1     & 0.335       \\
    5   & 2                                        & 0.790                                    & 2                                        & 0.614       & 2     & 0.255       \\
    6   & 3                                        & 0.754                                    & 3                                        & 0.557       & 2     & 0.194       \\
    7   & 3                                        & 0.719                                    & 3                                        & 0.505       & 2     & 0.148       \\
    8   & 4                                        & 0.686                                    & 3                                        & 0.458       & 3     & 0.113       \\
    9   & 4                                        & 0.655                                    & 4                                        & 0.415       & 3     & 0.086       \\
    10  & 5                                        & 0.624                                    & 4                                        & 0.377       & 4     & 0.065       \\
    \bottomrule
  \end{tabular}
\end{table}

At low noise, $k^*$ grows steadily while $p_{500}^h$ decays gently, and the survival probability remains above $0.5$ throughout. At moderate noise, $k^*$ follows a similar trajectory but $p_{500}^h$ crosses below $0.5$ between $h = 7$ and $h = 8$, so that at the self-play planning horizon of $h = 5$ the survival probability is above the feasibility threshold but with a margin that the myopic analysis would not reveal. At high noise the two effects diverge most starkly: the myopic case yields $k^* = 0$, multi-step planning recovers $k^* = 2$ at $h = 5$, yet $p_{500}^h$ falls below $0.5$ already at $h = 3$. The forgiveness machinery is accurate but irrelevant because the opponent is more likely than not to have genuinely defected somewhere in the planning window. This is the central asymmetry: longer horizons make the agent better at telling noise from hostility but exponentially less likely to face cooperative intent in the first place.

The critical threshold $\tilde{\varepsilon}(h)$ is the noise level at which $p_\infty(\varepsilon)^h = 0.5$. For $h = 1$ this recovers the myopic condition. For $h > 1$ the two competing effects interact: evidence aggregation pushes the threshold upward while intent persistence pulls it downward. From the empirical priors, the threshold lies between $\varepsilon = 0.10$ and $\varepsilon = 0.15$ for $h = 5$, consistent with the phase transition in \cref{fig:scores,fig:coop-rate}.

\subsection*{Empirical Correspondence}

The empirical cooperation rates coincide closely with the cooperative survival probabilities at all three noise levels (\cref{tab:hstep-correspondence}).

\begin{table}[h]
  \centering
  \caption{Cooperative survival probability $p_{500}(\text{CC})^5$ versus empirical cooperation rate (POMDP, self-play, $h = 5$).}
  \label{tab:hstep-correspondence}
  \begin{tabular}{cccc}
    \toprule
    $\varepsilon$ & $p_{500}(\text{CC})^5$ & Empirical cooperation & Difference \\
    \midrule
    0.05          & 0.790                  & 0.80                  & $+0.010$   \\
    0.10          & 0.614                  & 0.60                  & $-0.014$   \\
    0.15          & 0.255                  & 0.25                  & $-0.005$   \\
    \bottomrule
  \end{tabular}
\end{table}

This is consistent with an interpretation in which the converged cooperation rate equals the probability that cooperative intent persists across the planning window: the agent cooperates when the cooperative hypothesis dominates its $h$-step lookahead, and the frequency with which this occurs is determined by $p_t(\text{CC})^h$. However, three co-monotonic data points are insufficient to distinguish this relationship from other monotonic functions of the same quantities, and verification at intermediate noise levels would be needed to establish whether the correspondence is structural.

\subsection*{The Role of Epistemic Probing}
\label{app:epistemic-probing}

\Cref{prop:hstep} assumes a fixed cooperative prior and characterises when multi-step evidence aggregation sustains cooperation. Two mechanisms operate beyond this bound.

\paragraph{Active restoration of the cooperative prior.}
The expected free energy's parameter information gain term provides a stabilising force that the proposition's fixed-prior assumption excludes. When the agent's belief is displaced from CC into an uncertain non-CC state, the Dirichlet counts for that state are sparse relative to the well-visited CC state. Cooperation from these uncertain states is therefore highly informative about the transition model, producing a large parameter information gain that favours cooperative actions precisely when the pragmatic case for cooperation is weakest. This epistemic incentive creates a recovery pathway that pure evidence aggregation cannot provide: it drives the agent to test whether the opponent will reciprocate from non-CC states, generating the cooperative transitions needed to restore the system toward CC. The explore-then-commit trajectory documented in \cref{sec:efe-decomp} reflects this mechanism during initial learning; near the noise threshold, the same mechanism operates in miniature each time a noisy observation displaces the belief state.

The distinction is that the proposition evaluates whether the agent should cooperate given its current model, while epistemic probing changes the model itself. The proposition provides a lower bound on the conditions under which cooperation is sustained; epistemic probing can extend cooperation into regions where the bound is not satisfied by restoring $p_t(\text{CC})$ after transient retaliatory episodes.

\paragraph{Stochastic variance in the cooperative prior.}
The proposition treats $p_t(\text{CC})$ as converged, but the Dirichlet counts fluctuate as evidence accumulates. A transient upward fluctuation can temporarily satisfy the forgiveness condition even when the converged value would not, and a downward fluctuation can trigger retaliation even when the converged value should sustain cooperation. This variance is visible in the wide confidence intervals at intermediate noise levels in \cref{fig:scores,fig:coop-rate} and contributes to the bimodal character of the phase transition across seeds. A stochastic analysis treating $p_t$ as a random variable is left to future work.

\paragraph{Why the basin shrinks with noise.}
Both mechanisms lose effectiveness as $\varepsilon$ increases. The CC forgiveness margin narrows, so the planning horizon and epistemic drive must compensate for stronger hostile posteriors more frequently. The rate of challenging observations increases ($1-(1-\varepsilon)^2$), so the basin is tested more often. At $\varepsilon = 0.15$, CC has crossed into the hostile regime under myopic attribution, and neither the planning horizon nor the epistemic incentive can sustain cooperation from a state where the evidence so strongly favours retaliation. \Cref{tab:multistep-verification} confirms that the failure is not one of discriminability but of cooperative survival.

\subsection*{Why the MDP Exhibits No Sharp Threshold}

The attribution boundary characterises when Bayesian updating over the hidden intention variable flips from forgiving to hostile. The MDP has no latent intention layer and treats executed actions as states, so neither \cref{eq:attribution} nor \cref{prop:hstep} applies. Without a latent intention space there is no cooperative hypothesis $H_C$ to evaluate.

\begin{table}[h]
  \centering
  \caption{MDP cooperative priors (self-play, $t = 500$). These reflect transition regularities in executed actions, not latent intentions.}
  \label{tab:mdp-priors}
  \begin{tabular}{ccccc}
    \toprule
    $\varepsilon$ & $p_{500}(\text{CC})$ & $p_{500}(\text{CD})$ & $p_{500}(\text{DC})$ & $p_{500}(\text{DD})$ \\
    \midrule
    0.05          & 0.9129               & 0.5683               & 0.2186               & 0.8839               \\
    0.10          & 0.8294               & 0.7887               & 0.3002               & 0.7311               \\
    0.15          & 0.7795               & 0.7789               & 0.5112               & 0.7537               \\
    \bottomrule
  \end{tabular}
\end{table}

Despite lacking intention inference, the MDP achieves non-trivial cooperation through two features of active inference operating over its state space. First, B-matrix learning captures statistical regularities in executed transitions, even though these conflate intention and noise. Second, the permanently elevated epistemic drive (\cref{sec:efe-decomp}) provides persistent exploratory pressure that prevents the agent from locking into pure defection. The result is an uncommitted, drifting policy that accumulates cooperative episodes through exploration rather than conviction.

The MDP's B-matrix estimates decay smoothly with noise and cooperation declines approximately linearly with no phase transition (\cref{tab:mdp-priors}). The MDP never builds the high cooperative baseline that makes the POMDP's collapse dramatic. Its inability to perform intention inference simultaneously prevents it from reaching the POMDP's cooperative peak and prevents the correlated collapse that mutual intention inference produces. The POMDP's fragility at high noise arises not from inferior inference but from what accurate inference correctly discovers: that the opponent's cooperative intent cannot survive the noise-induced belief cascade.

\subsection*{Contrast with Fixed-Opponent Play}

The self-play collapse is specific to mutual inference. Against TFT, the focal agent's belief-driven retaliation does not corrupt the opponent's behaviour because TFT's policy is stationary regardless of the focal agent's belief dynamics. This dampens the cascade in two ways: TFT's consistent conditional cooperation provides a stable learning target that maintains high $p_t(\text{CC})$ in the Dirichlet counts, and a single retaliatory episode does not trigger a symmetric belief cascade because TFT has no beliefs to corrupt.

In self-play, both agents perform inference. When one agent's belief tips and it retaliates, the other agent correctly identifies genuine hostile intent and retaliates in turn. The spiral is not a failure of inference but a consequence of its accuracy. Against TFT, noise can only produce transient perturbations; in self-play, noise can trigger a self-reinforcing transition from cooperative to hostile equilibrium via the belief-intent pathway.

The $h$-step analysis locates this asymmetry precisely. Against a fixed opponent, $p_t(\text{CC})$ is determined by the opponent's stationary policy and is not degraded by the focal agent's belief dynamics, so the cooperative survival probability $p_t(\text{CC})^h$ remains high and stable. In self-play, $p_t(\text{CC})$ is a function of both agents' inference processes, and noise-induced retaliation by one agent degrades the other's Dirichlet priors. The same quantity that \cref{prop:hstep} identifies as governing cooperation feasibility is precisely the quantity that mutual inference destabilises.

\section{Prior Sensitivity}
\label{app:prior-sensitivity}

To evaluate sensitivity to the Dirichlet initialization, we compare three priors in self-play across three noise levels: a \emph{uniform} prior (the default), a \emph{reciprocity-favouring} prior (initialized toward conditional cooperation), and a \emph{defection-favouring} prior. All other settings match the main experiments (30 repetitions). \Cref{tab:prior-sensitivity} reports the mutual cooperation rates and mean per-turn scores.

\begin{table}[h]
  \centering
  \caption{Prior sensitivity in self-play: mutual cooperation rate and mean per-turn score for three Dirichlet initialisations across noise levels (30 repetitions). Per-turn score lies between the mutual-defection payoff ($1$) and the mutual-cooperation payoff ($3$).}
  \label{tab:prior-sensitivity}
  \begin{tabular}{l ccc ccc}
    \toprule
                         & \multicolumn{3}{c}{Cooperation rate} & \multicolumn{3}{c}{Score per turn}                                  \\
    \cmidrule(lr){2-4}\cmidrule(lr){5-7}
    Prior                & $\varepsilon{=}0.05$ & $0.10$ & $0.15$ & $\varepsilon{=}0.05$ & $0.10$ & $0.15$ \\
    \midrule
    Uniform              & 0.79 & 0.62 & 0.26 & 2.84 & 2.69 & 2.11 \\
    Reciprocity-favouring & 0.87 & 0.80 & 0.71 & 2.91 & 2.87 & 2.81 \\
    Defection-favouring  & 0.89 & 0.71 & 0.14 & 2.92 & 2.71 & 1.71 \\
    \bottomrule
  \end{tabular}
\end{table}

While the overall trends are robust, the prior's impact increases sharply near the critical noise threshold (\cref{sec:threshold}). At low noise ($\varepsilon$ = 0.05), all priors sustain high cooperation (0.79 to 0.89); these minor differences suggest the prior primarily affects convergence speed rather than the final outcome. At $\varepsilon$ = 0.10, the spread widens and the reciprocity-favouring prior shows a clear advantage. By $\varepsilon$ = 0.15, the prior becomes decisive: the reciprocity-favouring prior sustains cooperation (0.71), while the uniform and defection-favouring priors collapse (0.26 and 0.14).

Consistent with our threshold analysis, a prior aligned with the opponent's conditionally cooperative structure extends the noise range that supports cooperation. Specifically, it increases the initial cooperative survival probability, $p_t(CC)^h$, which delays the fixed-point collapse (\cref{eq:fixed-point}). Although the defection-favouring prior is competitive at low noise, its fragility near the threshold indicates that robust cooperation requires encoding reciprocity, not merely utilizing a strong prior.

\section{Notation and Correspondence to Bayesian ML}
\label{app:notation}

\Cref{tab:notation} relates the AIF notation used in the main text to its standard Bayesian ML and Reinforcement Learning (RL) counterparts. The two formalisms coincide: AIF's generative model is a controlled HMM, perception is variational inference, and the expected free energy is an exploration-aware planning objective.

\begin{table}[h]
  \centering
  \caption{Active-inference symbols and their standard probabilistic ML and RL analogues.}
  \label{tab:notation}
  \small
  \begin{tabular}{l p{4.5cm} p{5.5cm}}
    \toprule
    Symbol                      & Active-inference meaning                  & Probabilistic ML / RL analogue                                 \\
    \midrule
    \multicolumn{3}{c}{\textit{States \& Variables}} \\
    \midrule
    $s_t$                       & Hidden intention state                    & Latent state (HMM)                                             \\
    $o_t$                       & Executed action (observation)             & Observed emission                                              \\
    $\pi,\ a_t$                 & Policy / action                           & Action sequence / control input                                \\
    \midrule
    \multicolumn{3}{c}{\textit{Generative Model}} \\
    \midrule
    $\mathbf{A}$                & Observation model $P(o\mid s)$            & Emission / likelihood matrix                                   \\
    $\mathbf{B}$                & Transition model $P(s'\mid s,a)$          & Transition matrix                                              \\
    $\tilde{P}(o)$              & Preferred observations (prior preference) & Target / goal distribution (reward as log-prior)               \\
    \midrule
    \multicolumn{3}{c}{\textit{Inference \& Planning Objectives}} \\
    \midrule
    $Q(\cdot)$                  & Approximate posterior                     & Variational posterior                                          \\
    $F$                         & Variational free energy                   & Negative ELBO                                                  \\
    $G(\pi)$                    & Expected free energy                      & Planning objective (expected reward $+$ exploration bonus)     \\
    Epistemic value / info gain & Uncertainty-reducing drive                & Expected information gain / intrinsic exploration bonus        \\
    Pragmatic value             & Preference-seeking drive                  & Expected log-likelihood of preferred outcomes ($\approx$ extrinsic reward) \\
    \midrule
    \multicolumn{3}{c}{\textit{Online Learning}} \\
    \midrule
    $N^a_{ss'}$                 & Dirichlet transition counts               & Sufficient statistics / pseudo-counts                          \\
    $\alpha_0$                  & Dirichlet concentration                   & Smoothing (add-one) prior strength                             \\
    \bottomrule
  \end{tabular}
\end{table}

\end{document}